\documentclass[12pt]{interact}
\usepackage{setspace}
\usepackage{epstopdf}
\usepackage{subfigure}

\usepackage{natbib}
\bibpunct[, ]{(}{)}{;}{a}{}{,}
\renewcommand\bibfont{\fontsize{10}{12}\selectfont}

\usepackage{amsmath,amssymb,amsfonts}
\usepackage{algorithmic}
\usepackage{graphicx}
\usepackage{textcomp}
\usepackage{xcolor}
\usepackage{url}

\usepackage{amsthm}
\newtheorem{theorem}{Theorem}[section]
\newtheorem{definition}[theorem]{Definition}

\newtheorem{example}[theorem]{Example}
\newtheorem{remark}[theorem]{Remark}
\def\*#1{\mathbf{#1}}
\def\~#1{\boldsymbol{#1}}

\usepackage{booktabs}
\usepackage{subfiles}
\usepackage{colortbl}
\usepackage{subfigure}
\usepackage{graphicx}
\usepackage{caption}

\usepackage[most]{tcolorbox}

\tcbset{colback=gray!10, colframe=black, left=4mm, right=4mm, top=2mm, bottom=2mm}

\begin{document}


\title{Shapley in Context: Explaining Financial Language with Domain Expertise}

\author{
\name{Dangxing Chen \textsuperscript{a}\thanks{Corresponding author. Email: dangxing.chen@gmail.com} and Pengzhan Guo \textsuperscript{b}} 
\affil{\textsuperscript{a}Zu Chongzhi Center, Duke Kunshan University, Kunshan, Jiangsu, CN; \\ \textsuperscript{b}Digital Innovation Research Center, Duke Kunshan University, Kunshan, Jiangsu, CN}
}

\maketitle

\begin{abstract}
In recent years, large language models have achieved remarkable success and have seen growing adoption in financial applications. At the same time, explainability remains critical in finance, a domain characterized by high stakes and strict regulatory requirements. Although numerous methods have been proposed to explain black-box machine learning models, the majority of these approaches are designed for general-purpose tasks and do not incorporate domain-specific knowledge. In this work, we study the explainability of financial textual data modeled by large language models through the lens of the Shapley value. Specifically, we investigate whether Shapley-based attributions align with established financial domain knowledge. Through rigorous theoretical analysis and extensive empirical evaluations, we demonstrate that Shapley values can yield explanations that are consistent with financial reasoning and can offer meaningful insights into the model’s behavior in text-based financial applications.

\end{abstract}


\section{Introduction}

In recent decades, artificial intelligence (AI) and machine learning (ML) have achieved remarkable success across a wide range of fields. Among these advances, \emph{large language models} (LLMs) have attracted increasing attention, particularly following the release of ChatGPT in 2022. In the financial domain, LLMs are beginning to play an increasingly significant role and hold the potential to transform various aspects of finance through natural language processing (NLP), facilitating improved data generation, analysis, decision-making, and risk management. For instance, \cite{huang2023finbert} demonstrates how FinBERT, a finance-oriented LLM, enhances several key tasks in financial information processing; \cite{wu2024exploring} explores the use of LLMs in understanding behavior graphs and improving recommendation systems in online recruitment, including support for out-of-distribution scenarios; \cite{wu2023bloomberggpt} introduces BloombergGPT, a domain-specific LLM for finance that surpasses general-purpose models on financial tasks by utilizing specialized training data, while remaining competitive on broader NLP benchmarks; \cite{zhang2023enhancing} presents FinGPT, an LLM developed by the AI4Finance Foundation for real-time financial analysis through continual learning from public sources such as news, tweets, and Securities and Exchange Commission filings. A comprehensive survey of this emerging area can be found in \cite{li2023large}.

While LLMs have demonstrated success across a wide range of tasks, this often comes at the cost of reduced transparency and explainability. In high-stakes domains such as the financial industry, explainability is not merely desirable but essential. The \emph{Model Risk Management Handbook}\footnote{\url{https://www.occ.treas.gov/publications-and-resources/publications/comptrollers-handbook/files/model-risk-management/index-model-risk-management.html}}, published by the Office of the Comptroller of the Currency, emphasizes that ML models must be interpretable and their predictions explained in a clear and concise manner when deployed in practice. This requirement has driven significant research interest in the field of explainable machine learning, as evidenced by works such as \cite{lundberg2017unified, ribeiro2016should, horel2020significance, sundararajan2017axiomatic}.

Our work centers on the attribution problem in explainable ML models, that is, how one can assign a value to each input feature to reflect its contribution to the model’s prediction. {We focus on post-hoc explanations, where the model is fixed and has already been trained.} An influential approach to attribution is the Shapley value framework, initially proposed by \cite{shapley1953value} in game theory, and subsequently extended to ML applications by \cite{vstrumbelj2014explaining,lundberg2017unified}.
 One of the appealing aspects of this approach is its axiomatic foundation: with a small set of intuitive axioms, a unique attribution method can be derived. The rigorous mathematical foundation offers a solid framework, enabling practitioners to interpret complex black-box models in a fair and trustworthy manner. Building on this foundation, a wide range of Shapley-based attribution variants have since been proposed, see for example a comparison of different methods in the Shapley value family by  \cite{sundararajan2020many}.

Despite the success of the axiomatic approach, prior work has largely focused on attribution methods for general-purpose models without incorporating domain-specific knowledge \citep{sundararajan2020many, mosca2022shap}. In contrast, the value of domain expertise has been well recognized throughout the history of scientific inquiry. Recent research suggests that integrating domain knowledge into the explanation of ML models can significantly enhance their interpretability and relevance \citep{chen2025explaining, chen2024attribution, chen2024groups,chen2023interpretable}. Motivated by this, we turn our attention to the explainability of models applied to financial textual data, a less traditional yet increasingly important problem in the financial domain.

In this work, we ask the question: \textbf{Do Shapley value-based attributions provide consistent explanations for LLMs applied to financial text data?} To address this question, we propose a set of axioms that capture our expectations when interpreting LLM behavior. 
For instance, in sentiment analysis, we expect that a positive word should receive a positive attribution. Moreover, a more positive word should yield a higher attribution value. In more complex cases, we may wish to compare the relative importance of different words within the same statement, or even across different statements.
Through rigorous analysis, we show that the baseline Shapley value satisfies a set of domain knowledge-inspired axioms, enabling it to faithfully capture the importance of textual features. As a result, we can confidently use these attributions to interpret LLM outputs on financial texts.
We validate our findings using several LLM tasks, illustrating how the proposed axioms are respected in practice and how baseline Shapley values provide meaningful insights into model behavior.

Our work builds upon the recent contribution by \cite{chen2024attribution}, who investigate attribution methods in the context of asset pricing based on analytic functions. While their results offer valuable insights, the analytic-function-based foundation is not directly applicable to LLMs, which operate on unstructured data where features are neither naturally ordered nor associated with differentiable functions.
In this work, we extend their framework by considering \textit{partially ordered functions}, which are more suitable for modeling LLM behavior. Furthermore, we analyze the \textit{relative importance} of different features, both within a single model and across different models.
Our main contributions are summarized as follows:
\begin{itemize}
    \item We propose a set of domain knowledge-inspired axioms tailored to the structure of LLMs applied to financial text data.
    \item We demonstrate, through rigorous theoretical analysis, that the baseline Shapley value satisfies these axioms and thus provides reliable feature attributions.
    \item We empirically validate our theoretical findings on multiple LLM examples and show how Shapley-based attributions yield interpretable insights into financial language modeling.
\end{itemize}


The paper is organized as follows. Section 2 introduces the preliminaries on attribution methods and the Shapley value. In Section 3, we analyze the Shapley value for financial textual data through the lens of domain knowledge-inspired axioms. Section 4 {presents a large-scale experiment analyzing the risk attributions of the risk factors section in Form~10-K filings
}. Finally, Section 5 concludes the paper.

\section{Preliminaries}
\label{prerequisites}
We are interested in the classification and regression when the inputs are textual data. As an example of sentiment analysis, we might be interested in analyzing whether the following statement is positive, neutral, or negative:
\begin{align} \label{eq:NLP_eg1}
    \text{``The company reported strong earnings."}
\end{align}
The LLM model can then be used to generate a sentiment score, which in this instance should be very positive. Here, features may consist of words or \( N \)-grams. For instance, we might be interested in the attributions of phrases like ``the company reported'', ``strong'', or ``earnings''.

For problem setup, assume we have $n$ features. Given a vocabulary $\mathbb{V}$, consider a sequence of words (or $N$-gram) $\*x = (x_1 \dots x_n) \in \mathbb{X}$, where $x_i \in \mathbb{X}_i \subseteq \mathbb{V}$, $\mathbb{X}_i$ contains all possible choices for $x_i$, and $\mathbb{V}$ contains all possible phrase. If the output is a numerical value, say the sentiment score, then we are interested in the function of the form
\begin{align*}
    f: \mathbb{X}_1 \times \dots \times \mathbb{X}_n \rightarrow \mathbb{R}.
\end{align*}
Here, we interpret $f$ as an LLM model. As an example, $f$ may represent an LLM model that computes the expected sentiment of a given sentence. {We emphasize that our focus is on post-hoc explanations. That is, the model $f$ is already trained and fixed, and we do not modify it in any way. Instead, we aim to attribute the output of $f$ to its input features or text components, analyzing how different parts contribute to the prediction. This approach allows us to interpret and understand the model’s behavior without interfering with its training or performance.} For multiclass classification problem, we have a vector function of the form
\begin{align*}
    \*f: \mathbb{X}_1 \times \dots \times \mathbb{X}_n \rightarrow \mathbb{R}^n.
\end{align*}
For example, in the sector classification task where the goal is to predict the appropriate sector (e.g., finance, technology, etc.) based on a given text, the function \( f_i \) represents the probability that the text belongs to the \( i \)th class. We denote the family of such functions as $\mathcal{F}$ and we require that $f \in \mathcal{F}$ and $\*f \in \mathcal{F}$. 
Note that in this case, the elements in $\mathbb{X}_i$ are not necessarily ordered. However, in certain scenarios, a partial ordering may be defined over them. As in Example~\eqref{eq:NLP_eg1}, different choices for \(x_2\) = ``strong" are possible. For instance, consider the set \(\mathbb{X}_2 = \{``\text{weak}", ``\text{expected}", ``\text{strong}",``\text{solid}"\}\). While it may not be possible to directly compare ``strong'' and ``solid'', ``strong'' is clearly a stronger tone than ``weak''. This partial ordering allows us to analyze the properties of the function.

\subsection{Attribution Methods}
Following \cite{lundstrom2022rigorous}, we call the point of interest $\overline{\*x}$ to explain as an explicand, $\*x'$ a baseline, and $\*x$ as a general function input. The Baseline Attribution Method that interprets features' importance is defined below.

\begin{definition}[Baseline Attribution Method (BAM)] \label{def:BAM}
    Given an explicand $\overline{\mathbf{x}} \in \mathbb{X}$, a baseline $\mathbf{x}' \in \mathbb{X}$, and a function $f \in \mathcal{F}$, a baseline attribution method (BAM) is defined as any function of the form 
\[
\mathbf{A}(\overline{\mathbf{x}}, \mathbf{x}', f): \mathbb{X} \times \mathbb{X} \times \mathcal{F} \rightarrow \mathbb{R}^n.
\]
For simplicity, we may also write $\mathbf{A}$, $\mathbf{A}(\overline{\mathbf{x}})$, or $\mathbf{A}(f)$ when the other components are understood from context. Similarly, we use $A_i$, $A_i(\overline{\mathbf{x}})$, or $A_i(f)$ to denote the $i$-th component of the attribution vector.
\end{definition}

The definition of BAMs is broad, as it encompasses a wide range of potential methods. 
\cite{shapley1951notes} proposed Shapley value, which takes as input a characteristic function $v:2^N \rightarrow \mathbb{R}$, which produces attributions $\text{SH}_i$ for each player $i \in N$ that add up to $v(N)$, where $N = \{1, \dots, n\}$ and $2^N$ denotes the power set of $N$.

\begin{definition}[Shapley value]
    The Shapley value of a player $i$ is given by:
\begin{align*}
    \text{SH}_i(v) := \sum_{S \subseteq N \backslash \{i\}} \frac{|S|! (|N|-|S|-1)!}{|N|!} (v(S \cup \{i\}) - v(S)).
\end{align*}
\end{definition}
By employing an axiomatic approach, the Shapley value offers a rigorous and fair allocation of attributions based on the given characteristic function, thereby enabling users to adopt it with greater confidence. The result is formally summarized in the following theorem by \cite{shapley1953value}.

\begin{theorem} \label{thm:SH_uniq}
    Given the set $N = \{1, \dots, n\}$ and the characteristic function $v:2^N \rightarrow \mathbb{R}$, Shapley value is the unique method that preserves the following axioms:
    \begin{itemize}
        \item Linearity: Given two characteristic functions $v$ and $w$, $\text{SH}_i(v+w) = \text{SH}_i(v) + \text{SH}_i(w)$ for all $i$. 
        \item Completeness: $\sum_{i \in N} \text{SH}_i = v(N) - v(\emptyset)$.
        \item Symmetry: If $v(S \cup \{i\}) = v(S \cup \{j\})$ for all $S \subset N$ that doesn't contain $i$ and $j$, then $\text{SH}_i = \text{SH}_j$. 
        \item Dummy: If $v(S \cup \{i\}) = v(S)$ for all $S \subset N$ that doesn't contain $i$, then $\text{SH}_i=0$. 
    \end{itemize}
\end{theorem}

The axioms underlying the Shapley value ensure that the resulting attributions represent fair allocations that align with our requirements. The \textit{linearity} axiom guarantees that the Shapley value is computed in a transparent and interpretable manner. The \textit{completeness} axiom ensures that the sum of all attributions equals the function value, thereby making the attributions meaningful. The \textit{symmetry} axiom states that if two features contribute equally across all subsets, they must receive equal attributions, reinforcing fairness. Finally, the \textit{dummy} axiom ensures that any feature with no contribution to the prediction receives zero attribution.

While Shapley value is uniquely determined, there are different choices of characteristic functions. We focus on the \textit{Baseline Shapley value (BShap)} proposed by \cite{sundararajan2020many}, in which
\begin{align*}
    v(S) := f(\overline{\*x}_S; \*x'_{N \backslash S}).
\end{align*}
That is, baseline values replace the feature's absence. We denote BShap's attribution by $\textbf{BS}(\overline{\*x},\*x',f)$ and denote the $i$th attribution by $\text{BS}_i(\overline{\*x},\*x',f)$. For simplicity, we may also write $\textbf{BS}$, $\textbf{BS}(\overline{\mathbf{x}})$, or $\textbf{BS}(f)$ when the other components are understood from context. Similarly, we use $\text{BS}_i$, $\text{BS}_i(\overline{\mathbf{x}})$, or $\text{BS}_i(f)$ to denote the $i$-th component of the attribution vector. As in Example \eqref{eq:NLP_eg1}, we can consider $x_1' =$``the company reported", $x_2'=$`` ", and $x_3'=$``financial statement". Note that for this selection, we are primarily interested in the attributions for $x_2$ and $x_3$, as $x_1$ is believed to be less relevant for sentiment.
As a result, we consider a baseline text
\begin{align*}
    ``\text{The company reported financial statement.}"
\end{align*}
Now suppose $S = \{2\}$, then the characteristic function calculates 
\begin{align*}
    v(S) = f(``\text{The company reported strong financial statement.}")
\end{align*}
In this case, since $\overline{x}_1=x_1'$, it reduces to a two-dimensional attribution problem with $\text{BS}_1=0$. The attribution for ``strong" is then calculated as 
\begin{align*}
    \text{BS}_2 &= \frac{f(``\text{The company reported strong financial statement.}")}{2} \\
    & - \frac{f(``\text{The company reported financial statement.}")}{2} \\
    &+ \frac{f(``\text{The company reported strong earnings.}")}{2} \\
    &- \frac{f(``\text{The company reported earnings.}")}{2}.
\end{align*}
In other words, it takes the average difference of cases when ``strong" is and is not specified in the sentence. 
 We focus on the BShap since it has better theoretical properties than other variants of Shapley values in terms of preserving axioms, as discussed in \cite{sundararajan2020many}. As a result, we are more confident about using it for sectors with high stakes. 

 \begin{remark}
     In this study, we primarily focus on the baseline Shapley value rather than other popular attribution methods. This choice is motivated by our desire to use attribution methods that preserve the essential axioms for financial textual data, ensuring meaningful interpretation. Other methods may lack these desirable properties. For example, Integrated Gradients (IG) proposed by \cite{sundararajan2017axiomatic} is a notable alternative that has demonstrated strong interpretability in certain financial applications, such as asset pricing, as discussed in \cite{chen2024attribution}. Although IG can be applied to textual explanations, it may not be ideal for our purpose of producing consistent explanations aligned with financial knowledge. This limitation arises because IG requires computing derivatives of continuous features, assuming that the input is continuous and differentiable, conditions not naturally met by textual data. While it is possible to embed words into continuous vectors and apply IG, such transformations may not preserve the relative relationships among words, thereby potentially violating our desired axioms. In contrast, the baseline Shapley value inherently operates on discrete features, making it more suitable for our objectives.

 \end{remark}

\subsection{Multiclass Classification}

For multiclass classification, we adopt a one-vs-rest approach. Given a probability vector $\mathbf{f}$ produced by the model, we treat each component function $f_i$ individually and compute the Shapley attribution $\text{BS}_k(\overline{\mathbf{x}}, \mathbf{x}', f_i)$ for the $k$-th feature. Variations of this method are also possible. For instance, to understand the attribution that drives the prediction from class $j$ to class $i$, we can consider the difference between probability functions: $\text{BS}_k(\overline{\mathbf{x}}, \mathbf{x}', f_i - f_j)$. Due to the linearity property of the Shapley value, this difference can be expressed as:
\[
\text{BS}_k(\overline{\mathbf{x}}, \mathbf{x}', f_i - f_j) = \text{BS}_k(\overline{\mathbf{x}}, \mathbf{x}', f_i) - \text{BS}_k(\overline{\mathbf{x}}, \mathbf{x}', f_j).
\]

\section{Analyze Shapley Value for Financial Textual Data} \label{sec:Shap_theory}

When explaining LLMs in financial applications, it is crucial that the generated explanations align with established financial domain knowledge. To ensure this alignment, we introduce domain knowledge-inspired axioms that capture key financial principles and constraints. These axioms serve as formal criteria that any faithful explanation should satisfy, guiding the interpretation of model behavior in a way that is both meaningful and trustworthy within the financial context. Correspondingly, we provide an analysis of the Shapley value in this context, along with illustrative examples, to demonstrate how the proposed axioms influence and constrain the resulting explanations. The central principle behind our axiomatic framework is as follows: if a model exhibits certain behavior for which we have a clear expectation regarding how feature attributions should respond, then the attribution method must produce results consistent with that expectation. In this way, if an attribution method preserves the required axioms and performs reliably under these well-understood conditions, we can place greater trust in its outputs, not only in idealized scenarios but also, more importantly, in more complex and less structured real-world cases where such conditions may not be fully satisfied. All proofs in this Section are left in Appendix~\ref{sec:proof}.

\noindent \textbf{FinBERT}: To better illustrate our idea, we will use the FinBERT\footnote{\url{https://huggingface.co/ProsusAI/finbert}} \citep{araci2019finbert} as an example. FinBERT is a BERT-based language model tailored for the financial domain through fine-tuning on finance-specific texts. It is optimized for sentiment analysis and other NLP tasks in finance, offering greater accuracy than general-purpose models. Given a financial input, FinBERT outputs the probabilities of the text being positive, neutral, or negative. For simplicity, we focus on the expected sentiment score by computing the weighted average of these probabilities.
\begin{align*}
    \mathbb{E}[\text{sentiment}] = p_{\text{positive}} \cdot 1 + p_{\text{neutral}} \cdot 0 + p_{\text{negative}} \cdot -1.
\end{align*}
As a result, a positive expected sentiment score indicates a positive tone, and vice versa. The output score ranges from $-1$ to $1$.
Using FinBERT, we provide examples demonstrating how BShap can be used to explain attributions, illustrating how axioms are reflected in these explanations to enhance conceptual soundness.

\subsection{Axioms on First-order Main Effects}
Monotonicity is often employed to represent first-order effects of domain knowledge and is widely observed in financial contexts \citep{chen2023address}. For instance, a credit score is expected to decrease as the number of past-due payments increases.
 While monotonicity is typically defined globally for continuous functions, this assumption does not hold in our setting. Nevertheless, partial comparisons can still be made in many cases. In this section, we investigate the impact of the partially ordered version of the monotonicity axiom on baseline attribution methods. We begin by analyzing monotonicity with respect to a single feature, a property referred to as \emph{individual monotonicity}.
 Suppose $\alpha$ is the index of the feature we are interested in and $\neg$ its complement, then we partition the input $\*x$ into $\*x = (x_{\alpha},\*x_{\neg})$. The effect of individual monotonicity can then be captured by the following axiom. Intuitively, this axiom imposes a constraint on the sign of the attributions: if a feature consistently contributes positively to the prediction across all contexts, then it must receive a positive attribution.

\begin{definition}[\textbf{Average Individual Monotonicity Axiom}]
    Given an explicand $\overline{\*x} \in \mathbb{X}$, a baseline $\*x' \in \mathbb{X}$, a function $f \in \mathcal{F}$, and we partition $\*x = (x_{\alpha},\*x_{\neg})$. If $f(\overline{x}_{\alpha},\*x_{\neg}) \geq f(x_{\alpha}',\*x_{\neg})$ for all $\*x_{\neg}$, then
    \begin{align*}
        A_{\alpha}(\overline{\*x},\*x',f) \geq 0.
    \end{align*}
\end{definition}

\begin{theorem} \label{thm:AIM}
    BShap preserves the average individual monotonicity axiom. 
\end{theorem}

\begin{example} \label{eg:AIM}
    We consider explaining the sentiment score of the text
    \begin{align}
    ``\text{The company reported strong earnings for this quarter.}"
    \end{align}
    The expected sentiment score is approximately 0.94, aligning well with our intuition, as the statement is clearly very positive.
    Since the phrase ``the company reported'' conveys limited information on its own, we are not interested in its attributions.
    We are interested in attributions for 
    \begin{align*}
    \overline{x}_1 &= \text{``strong"}, \\
    \overline{x}_2 &= \text{``earnings"}, \\
    \overline{x}_3 &= \text{``for this quarter"}.
    \end{align*}
    For the baseline, we consider the statement
    \begin{align*}
    ``\text{The company reported financial statement for this year.}"
    \end{align*}
    The expected sentiment score is approximately -0.04, indicating that the models interpret the baseline statement as neutral. Correspondingly, the baselines are
    \begin{align*}
        x_1' &= \text{`` "}, \\ 
        x_2' &= \text{``financial statement"}, \\
        x_3' &= \text{``for this year"}.     
    \end{align*}
    We expect some monotonicity property for $x_1$. The word ``strong'' is generally regarded as highly positive in financial contexts. Therefore, we expect it to consistently contribute positively relative to a neutral baseline. As a result, the attribution for ``strong'' should be positive. By applying BShap, we obtain:
    \begin{align*}
        \mathbf{BS} = 
        \begin{bmatrix} 0.95 & 0.05 & -0.02 \end{bmatrix}^T.
    \end{align*}
    This result aligns with our intuition, as the positive sentiment of the sentence is primarily driven by the word ``strong,'' which carries the largest attribution. Furthermore, according to the average individual monotonicity axiom, the attribution of ``strong" is indeed positive.
\end{example}

In addition to the sign of an attribution, we are often interested in whether the attribution increases when the tone of a word becomes stronger. In other words, the axiom asks whether the relative importance of changes in individual features is faithfully reflected in their attributions.
 This intuition is captured by the following axiom.

\begin{definition}[\textbf{Demand Individual Monotonicity Axiom}] \label{thm:DIM}
    Given explicands $\overline{\*x},\*x^* \in \mathbb{X}$, a baseline $\*x' \in \mathbb{X}$, a function $f \in \mathcal{F}$, and we partition $\*x = (x_{\alpha},\*x_{\neg})$. If $f(\overline{x}_{\alpha},\*x_{\neg}) \geq f(x_{\alpha}^*,\*x_{\neg})$ for all $\*x_{\neg}$, then
    \begin{align*}
        A_{\alpha}((\overline{x}_{\alpha},\overline{\*x}_{\neg}),\*x',f) \geq A_{\alpha}((x^*_{\alpha},\overline{\*x}_{\neg}),\*x',f).
    \end{align*}
\end{definition}

\begin{theorem}
    BShap preserves the demand individual monotonicity axiom. 
\end{theorem}

\begin{example} \label{eg:DIM}
Following Example~\ref{eg:AIM}, now consider a similar statement:
\begin{align}
    \text{``The company reported stable earnings for this quarter."}
\end{align}
We consider the same setup of explicands and baseline, except that in this example $\overline{x}_1 = \text{``stable"}$.
The expected sentiment score is approximately 0.87, slightly lower than 0.94 in Example~\ref{eg:AIM}. This is reasonable since the statement remains positive but is less emphatic.

Applying BShap, we obtain:
\begin{align*}
    \mathbf{BS} = 
    \begin{bmatrix} 0.57 & 0.34 & -0.00 \end{bmatrix}^T.
\end{align*}
As expected, the word ``stable'' contributes the most. However, unlike in Example~\ref{eg:AIM}, the word ``earnings'' gains greater importance. This makes sense because ``stable'' by itself is somewhat vague; without explicitly referencing ``earnings,'' the tone tends toward neutral. Thus, there is a complex interaction between the words ``stable'' and ``earnings.''

We further look into this result. The average individual monotonicity guarantees that the word ``stable" receives a positive attribution. Furthermore, since the word ``strong" is clearly more positive than the word ``stable", the attribution of ``strong" is larger, due to the demand individual monotonicity axiom. 

\end{example}

\subsection{Axioms on Second-order Main Effects}

Second-order derivatives may provide additional insights into the model's behavior. The diminishing marginal effect is one of the most common phenomena \citep{gupta2018diminishing,pya2015shape}. As a simple example, consider that people are generally happy when given an apple to eat. While receiving more apples would typically make them happier, the additional increase in happiness tends to diminish. The following axiom is introduced to capture this phenomenon in the context of partially ordered features. As an intuitive motivation, we consider discrete analogues of first-order and second-order derivatives to characterize the behavior of attribution methods.

\begin{definition}[\textbf{Diminishing Marginal Effect Axiom}]
    Given explicands $\overline{\*x},\*x^* \in \mathbb{X}$, a baseline $\*x' \in \mathbb{X}$, a function $f \in \mathcal{F}$, and we partition $\*x = (x_{\alpha},\*x_{\neg})$. If $f(\overline{x}_{\alpha},\*x_{\neg}) \geq f(x_{\alpha}^*,\*x_{\neg}) \geq f(x_{\alpha}',\*x_{\neg})$ and $f(x_{\alpha}',\*x_{\neg}) - 2f(x^*_{\alpha},\*x_{\neg}) + f(\overline{x}_{\alpha},\*x_{\neg}) \leq 0$ for all $\*x_{\neg}$, then 
    \begin{align*}
        0 & \leq A_{\alpha} ((\overline{x}_{\alpha},\overline{\*x}_{\neg}),\*x',f) - A_{\alpha}((x^*_{\alpha},\overline{\*x}_{\neg}),\*x',f) \\
        &\leq A_{\alpha}((x^*_{\alpha},\overline{\*x}_{\neg}),\*x',f)-A_{\alpha}((x_{\alpha}',\overline{\*x}_{\neg}),\*x',f).
    \end{align*}
\end{definition}

\begin{theorem} \label{thm:DME}
    BShap preserves the diminishing marginal effect axiom.
\end{theorem}

\begin{example}
    Following Example~\ref{eg:AIM} and \ref{eg:DIM}, recall the corresponding attributions for ``strong" and ``stable" are 0.94 and 0.57, respectively. Starting from a neutral expression `` " to the word ``stable", the text sentiment becomes significantly more positive, resulting in a large positive attribution assigned to ``stable". However, since ``stable" already conveys a strong positive sentiment, the marginal increase in attribution when moving from ``stable" to ``strong" is expected to be smaller due to the diminishing marginal effect. The diminishing marginal effect axiom ensures that this property is indeed preserved in the attributions and only produces a marginal increment of 0.37.

\end{example}

Analogous to the diminishing marginal effect, an increasing marginal effect may also arise in certain contexts. In contrast to the former, the increasing marginal effect implies that each successive increment results in a larger and larger change in attribution. The following axiom is introduced to capture this behavior.

\begin{definition}[\textbf{Increasing Marginal Effect Axiom}] \label{def:IME_axiom}
    Given explicands $\overline{\*x},\*x^* \in \mathbb{X}$, a baseline $\*x' \in \mathbb{X}$, a function $f \in \mathcal{F}$, and we partition $\*x = (x_{\alpha},\*x_{\neg})$. If $f(\overline{x}_{\alpha},\*x_{\neg}) \geq f(x_{\alpha}^*,\*x_{\neg}) \geq f(x_{\alpha}',\*x_{\neg})$ and $f(x_{\alpha}',\*x_{\neg}) - 2f(x^*_{\alpha},\*x_{\neg}) + f(\overline{x}_{\alpha},\*x_{\neg}) \geq 0$ for all $\*x_{\neg}$, then 
    \begin{align*}
        0 & \leq A_{\alpha}((x^*_{\alpha},\overline{\*x}_{\neg}),\*x',f) - A_{\alpha}((x_{\alpha}',\overline{\*x}_{\neg}),\*x',f) \\
        & \leq 
        A_{\alpha} ((\overline{x}_{\alpha},\overline{\*x}_{\neg}),\*x',f) - A_{\alpha}((x^*_{\alpha},\overline{\*x}_{\neg}),\*x',f).
    \end{align*}
\end{definition}

\begin{theorem} \label{thm:IME}
    BShap preserves the increasing marginal effect axiom. 
\end{theorem}

\subsection{Axioms on Different Features}

So far, our focus has been on the effect of a single feature. We now turn to the analysis of relative importance between potentially different features. The concept of pairwise monotonicity was first introduced in \cite{gupta2020multidimensional, chen2022monotonic, chen2023address}, which addresses the inherent relative importance among different features. 
For instance, in credit scoring, consider two features: the number of past-due payments of less than three months and those of more than three months. Since longer delinquencies typically indicate greater risk, each additional past-due payment of longer duration should have a larger impact on the credit score. When generating explanations, such relative importance should be respected. To capture this, we introduce the following axiom, which states that if one feature consistently produces a larger impact on the model output than another when perturbed, then it must receive a greater attribution.

\begin{definition}[\textbf{Pairwise Monotonicity Axiom}]
    Given an explicand $\overline{\*x} \in \mathbb{X}$, a baseline $\*x' \in \mathbb{X}$, a function $f \in \mathcal{F}$, and we partition $\*x = (x_{\alpha},x_{\beta},\*x_{\neg})$. If $f(\overline{x}_{\alpha},x'_{\beta},\*x_{\neg}) \geq f(x_{\alpha}',\overline{x}_{\beta},\*x_{\neg})$ for all $\*x_{\neg}$, then 
    \begin{align*}
    A_{\alpha}(\overline{\*x},\*x',f) \geq A_{\beta}(\overline{\*x},\*x',f).
    \end{align*}
\end{definition}

\begin{theorem} \label{thm:PM}
    BShap preserves the pairwise monotonicity axiom. 
\end{theorem}

\begin{example}\label{eg:PM}
We consider the statement
\begin{align}
    \text{``Tesla reported stable gross margin and earnings."}
\end{align}
The expected sentiment score is approximately 0.91, indicating that it is a very positive news. We focus on the explained features 
\begin{align*}
    \overline{x}_1 &= \text{``stable''}, \\
    \overline{x}_2 &= \text{``gross margin''}, \\
    \overline{x}_3 &= \text{``earnings''}.  
\end{align*}
We use the baseline features as
\begin{align*}
    \overline{x}_1 &= \text{`` "}, \\
    \overline{x}_2 &= \text{``figure"}, \\
    \overline{x}_3 &= \text{``item"}.
\end{align*}
Therefore, the baseline statement is
\begin{align}
    \text{``Tesla reported figure and item."}
\end{align}
The expected score for the baseline is about \(-0.02\), which is neutral.

The calculated BShap values are
\begin{align*}
    \mathbf{BS} = \begin{bmatrix} 
    0.78 & 0.12 & 0.04 
    \end{bmatrix}^T.
\end{align*}
By the pairwise monotonicity axiom, we observe that ``gross margin'' is significantly more important than ``earnings'' for Tesla. This reflects Tesla’s current business context: the company is in a rapid growth and expansion phase, investing heavily in factories, R\&D, and new products. These investments lead to high operating expenses and capital expenditures, which can depress net earnings or even cause short-term losses.
Gross margin measures the difference between revenue and cost of goods sold, showing how much Tesla earns on each vehicle before overhead and other expenses. It indicates how efficiently Tesla manufactures and sells its products, a critical factor for scaling profitably.
Earnings include components such as stock-based compensation, interest, taxes, and one-time charges that may not reflect the core business performance. As a result, Tesla’s earnings can fluctuate substantially due to these factors, even if core product profitability improves.
Investors often view improving gross margins as evidence that Tesla is gaining pricing power, reducing production costs, or optimizing its product mix, all of which are important signals for sustainable, long-term profit growth.
\end{example}

\subsection{Axioms on Features across Different Models}

Pairwise monotonicity addresses the relative importance between different features within the same model. However, in many cases, we are also interested in assessing the relative importance of the same feature across different models, particularly to understand which model is more sensitive to a given input. For example, while companies regularly issue earnings announcements, the impact of these announcements can vary across firms. In such cases, we may be interested in understanding how the term ``earnings'' contributes differently to the predictions for different companies. In more general scenarios involving multiclass classification, it is natural to evaluate the relative importance of a specific feature across different classes.
Consider two models that share the same inputs, we are interested in comparing the attribution of a common feature $\overline{x}_{\alpha}$ across the two models to determine which model assigns it greater importance.
This notion is formalized in the following axiom.

\begin{definition}[\textbf{First-order Monotonic Dominance Axiom}]
    Given an explicand $\overline{\mathbf{x}} \in \mathbb{X}$, a baseline $\mathbf{x}' \in \mathbb{X}$, functions $f, g \in \mathcal{F}$, and we consider a partition of the input as $\mathbf{x} = (x_{\alpha}, \mathbf{x}_{\neg})$.
     If $f(\overline{x}_{\alpha},\*x_{\neg}) -f(x_{\alpha}',\*x_{\neg}) \geq g(\overline{x}_{\alpha},\*x_{\neg}) - g(x'_{\alpha},\*x_{\neg})$ for all  $\*x_{\neg}$, then
    \begin{align*}
        A_{\alpha} (\overline{\*x},\*x',f) \geq A_{\alpha} (\overline{\*x},\*x',g). 
    \end{align*}
\end{definition}

\begin{theorem}{\label{thm:FMD}}
    BShap preserves the first-order dominance axiom. 
\end{theorem}

\begin{example}\label{eg:SM}
We consider two statements:
\begin{align}
& \text{``Amazon reported stable earnings."}, \\
& \text{``Coca-Cola reported stable earnings."}.
\end{align}
These two statements convey the same content but pertain to two different companies. Our objective is to compare the effect of the term ``earnings'' across these two companies.
Their sentiment scores are 0.61 and 0.91, respectively, indicating positive news in both cases. Notably, the phrase ``stable earnings'' carries a more positive attribution for Coca-Cola. This is understandable since Coca-Cola operates in a mature, defensive industry where investors prioritize consistency, dividends, and low volatility over rapid growth. In this setting, ``stable earnings'' imply resilience, a desirable attribute. In contrast, the market expects technology companies to deliver growth, innovation, and outperformance.

Following this, suppose we consider 
\begin{align*}
\overline{x}_1 &= \text{``stable''}, \\ \overline{x}_2 &= \text{``earnings''}   
\end{align*}
 as explicands. At the same time, we take the baselines 
 \begin{align*}
 x_1' &= \text{`` ''}, \\
 x_2' &= \text{``financial statement''}.    
 \end{align*}
 We then obtain:
\begin{align*}
    \mathbf{BS}_{\text{Amazon}} = \begin{bmatrix} 0.38 \\ 0.36 \end{bmatrix}, \quad
    \mathbf{BS}_{\text{Coca-Cola}} = \begin{bmatrix} 0.52 \\ 0.61 \end{bmatrix}.
\end{align*}
By the First-order Monotonic Dominance property, it is expected that ``earnings'' carries greater importance for Coca-Cola than for Amazon.
\end{example}

Last, it is natural to ask whether we can compare different features across different models using the Shapley value. For instance, suppose both Amazon and Tesla report financial indicators such as gross margin and earnings. Can we meaningfully compare the importance of gross margin to Amazon with the importance of earnings to Tesla?
To reason about such comparisons, we follow the principle of first-order monotonic dominance, which leads us to the following axiom.

\begin{definition}[\textbf{Symmetric Monotonicity Axiom}]
    Given an explicand $\overline{\mathbf{x}}  \in \mathbb{X}$, a baseline $\mathbf{x}' \in \mathbb{X}$, functions $f, g \in \mathcal{F}$, and we partition $\*x = (x_{\alpha},x_{\beta},\*x_{\neg})$. If $f(\overline{x}_{\alpha},x_{\beta}',\*x_{\neg}) -f(x_{\alpha}',x_{\beta}',\*x_{\neg}) \geq g(x_{\alpha}',\overline{x}_{\beta},\*x_{\neg}) - g(x'_{\alpha},x'_{\beta},\*x_{\neg})$ for all $\*x_{\neg}$, then
    \begin{align*}
        A_{\alpha} (\overline{\*x}',\*x',f) \geq A_{\beta} (\overline{\*x},\*x',g). 
    \end{align*}
\end{definition}

Unfortunately, the answer is negative: we cannot directly compare the attributions of different features across different models using the Shapley value. To see this, consider a simple two-dimensional setting with functions $f(x_1, x_2)$ and $g(x_1, x_2)$. According to the BShap formulation, the attributions are given by:
\begin{align*}
    \text{BS}_{1}(f) &= \frac{f(\overline{x}_1,x_2') - f(x_1',x_2')}{2} + \frac{f(\overline{x}_1,\overline{x}_2) - f(x_1',\overline{x}_2)}{2}, \\
    \text{BS}_{2}(g) &= \frac{g(x_1',\overline{x}_2) - g(x_1',x_2')}{2} + \frac{g(\overline{x}_1,\overline{x}_2) - g(\overline{x}_1,x_2')}{2}.
\end{align*}
In the spirit of comparing the relative importance of different features, we might attempt to compare the first terms of each expression, such as $\frac{f(\overline{x}_1,x_2') - f(x_1',x_2')}{2}$ with $\frac{g(x_1',\overline{x}_2) - g(x_1',x_2')}{2}$, since they share the same baseline input $(x_1',x_2')$. However, the second terms, $\frac{f(\overline{x}_1,\overline{x}_2) - f(x_1',\overline{x}_2)}{2}$ and $\frac{g(\overline{x}_1,\overline{x}_2) - g(\overline{x}_1,x_2')}{2}$, cannot be meaningfully compared, as values of $f(\overline{x}_1,\overline{x}_2)$ and $g(\overline{x}_1,\overline{x}_2)$ matter. In other words, for discrete cases, the relative importance of different features across different models cannot be precisely specified.
This example highlights the inherent difficulty in comparing attributions of different features across different models, and thus we should not expect such comparisons to be valid in general.

\subsection{An Example of Multiclass Classification for Sector Classifications}

Here, we provide another example of sector classification \citep{dzuyo2025linking,li2024unified} as a representative example of general text classification and recommendation tasks. Sector classification plays a crucial role in organizing and understanding vast amounts of textual data by categorizing information into meaningful industry groups. For instance, identifying that a sentence describes the technology sector enables more targeted analysis, efficient information retrieval, and personalized recommendations. Accurate sector classification supports investors, analysts, and businesses in monitoring industry trends, making informed decisions, and tailoring services to sector-specific needs. Moreover, in applications such as financial news aggregation, automated reporting, and market research, distinguishing content by sector improves the relevance and precision of insights, ultimately enhancing the overall effectiveness of data-driven strategies.

Unlike sentiment analysis, which typically involves a few ordered classes (e.g., positive, neutral, negative for sentiment analysis), sector classification is more complex due to the presence of multiple unordered categories. Due to this, we focus on the relative importance of different features within the same class as well as the relative importance of a single feature across different sectors. This allows us to demonstrate the expressiveness and flexibility of the Shapley value in handling multiclass, unordered classification tasks.

\subsubsection{LLM Model}

We use \textit{BART-Large-MNLI (BLM)}\footnote{\url{https://huggingface.co/facebook/bart-large-mnli}}, a powerful transformer-based language model developed by Meta (formerly Facebook AI), built on the BART \citep{lewis2019bart} architecture and fine-tuned on the Multi-Genre Natural Language Inference (MNLI) dataset. As a sequence-to-sequence model, BART combines the strengths of BERT (for understanding) and GPT (for generation), and this version is specifically optimized for natural language inference tasks. The model takes a pair of sentences, typically a premise and a hypothesis, and predicts whether the hypothesis is entailed by, contradicts, or is neutral with respect to the premise. A particularly notable use of BART-Large-MNLI  is in zero-shot text classification, where classification labels are converted into natural language hypotheses, allowing the model to classify text without task-specific fine-tuning. This versatility has made it a popular choice for many NLP applications requiring flexible and generalizable language understanding.

\subsubsection{An Example of Explanations}

We consider the example of a financial technology (FinTech) company, Upstart. As a brief description, we use the following statement:
\begin{align}
    \text{``The company is a leading AI lending marketplace."}
\end{align}
For simplicity, we restrict our classification to six industry sectors: technology, finance, healthcare, consumer goods, energy, and others. Given Upstart's business model, we expect it to have substantial attribution in both the technology and finance sectors. Indeed, BLM assigns $56\%$ to technology and $33\%$ to finance, indicating that BLM identifies it as a FinTech company, with a stronger leaning toward the technology sector.

To analyze attributions, we focus on the words 
\begin{align*}
    \overline{x}_1 &= \text{``AI"}, \\
    \overline{x}_2 &= \text{``lending"},
\end{align*}
which are likely the key contributors to its classification as a FinTech firm. As a baseline, we set 
\begin{align*}
    x_1' &= \text{`` "}, \\
    x_2' &= \text{`` "},
\end{align*}
resulting in the baseline statement:
\begin{align*}
    \text{``The company is a leading marketplace."}
\end{align*}
Under this baseline, BLM assigns $24\%$ to technology and $16\%$ to finance.

We now apply BShap to assess the contributions of ``AI" and ``lending" to each sector. For this two-dimensional case, we provide further details regarding the function values associated with different subsets. For the technology sector, the function values are:
\begin{align*}
    f_{\text{tech}}(``\text{The company is a leading marketplace.}") &= 0.24, \\
    f_{\text{tech}}(``\text{The company is a leading AI marketplace.}") &= 0.82, \\
    f_{\text{tech}}(``\text{The company is a leading lending marketplace.}") &= 0.05, \\
    f_{\text{tech}}(``\text{The company is a leading AI lending marketplace.}") &= 0.56, \\
\end{align*}
Correspondingly, the BShap attributions are:
\begin{align*}
    \textbf{BS}_{\text{technology}} = \begin{bmatrix}
        0.55 & -0.22
    \end{bmatrix}^T.
\end{align*}
For the finance sector, the function values are: 
\begin{align*}
    f_{\text{finance}}(``\text{The company is a leading marketplace.}") &= 0.16, \\
    f_{\text{finance}}(``\text{The company is a leading AI marketplace.}") &= 0.03, \\
    f_{\text{finance}}(``\text{The company is a leading lending marketplace.}") &= 0.82, \\
    f_{\text{finance}}(``\text{The company is a leading AI lending marketplace.}") &= 0.33.
\end{align*}
Correspondingly, the BShap attributions are:
\begin{align*}
    \textbf{BS}_{\text{finance}} = \begin{bmatrix}
        -0.31 & 0.48
    \end{bmatrix}^T.
\end{align*}

For the technology class, the word ``AI'' contributes positively to the prediction, which aligns with the expectation that ``AI'' is clearly associated with technology companies. In contrast, the word ``lending'' reduces the prediction score for the technology class, as it is more commonly linked to finance companies. These sign patterns are consistent with the average individual monotonicity axiom. 
According to the pairwise monotonicity axiom, the contribution magnitude of ``AI'' exceeds that of ``lending''. This is reasonable, as the concept of ``lending'' often involves advanced technology, which mitigates its negative influence on the technology class prediction. 
A similar pattern is observed for the finance class. 
Additionally, according to the first-order monotonic dominance axiom, the importance of the same feature can be meaningfully compared across different models. We observe that ``AI'' has a stronger impact on the technology class, while ``lending'' contributes more significantly to the finance class, both observations aligning with our expectations. Supported by these axioms, we gain greater confidence in interpreting model outputs when using Shapley values.

\subsubsection{Applying the Shapley Value to Model Diagnostics}
We present another example illustrating the use of the Shapley value in a case where the model fails to perform as expected.
We consider a hedge fund, the Voleon Group, which extensively employs AI technologies in its trading activities. The following text describes Voleon:
\begin{center}
    ``\text{A quantitative investment management firm that applies artificial intelligence} \\
    \text{and machine learning to its trading strategies.}''
\end{center}
Clearly, Voleon group belongs to the finance sector. However, given its strong reliance on AI technologies, we would expect a nontrivial contribution from the technology sector. Surprisingly, the BLM model assigns $64\%$ probability to the technology sector and only $14\%$ to the finance sector—an evidently unreasonable result. Shapley values can help us uncover the reasoning behind this prediction.

For this text, we identify the following explicands:
\begin{align*}
    \overline{x}_1 &= \text{``quantitative investment management''}, \\
    \overline{x}_2 &= \text{``artificial intelligence and machine learning''}, \\
    \overline{x}_3 &= \text{``trading strategies''}.
\end{align*}
We consider these to be the key phrases determining the sector classification. For comparison, we construct a neutral baseline text:
\begin{center}
    ``A firm that applies service to its operations.''
\end{center}
Accordingly, the corresponding baseline components are:
\begin{align*}
    x_1' &= \text{`` ''}, \\
    x_2' &= \text{``service''}, \\
    x_3' &= \text{``operations.''}
\end{align*}
Applying BShap yields the following attributions:
\begin{align*}
    \textbf{BS}_{\text{tech}} &= \begin{bmatrix} -0.09 & 0.50 & 0.0175 \end{bmatrix}^T, \\
    \textbf{BS}_{\text{finance}} &= \begin{bmatrix} 0.11 & -0.20 & 0.06 \end{bmatrix}^T.
\end{align*}
From these BShap values, we observe that the phrase ``artificial intelligence and machine learning'' is the dominant factor driving the classification toward the technology sector and away from finance. While this is understandable, it is also notable that the phrases ``quantitative investment management'' and ``trading strategies'' contribute very little. These results suggest that BLM is relatively insensitive to some financial terminology. As such, fine-tuning the model on finance-specific textual data may be necessary to improve its accuracy in sector classification.
{For comparison, we also provide results from the state-of-the-art ChatGPT 5.2 in Appendix}~\ref{sec:sec_GPT}, {which demonstrate that ChatGPT produces more reasonable outcomes in this example.}

\subsection{{Summary}}

An overview of the domain knowledge-inspired axioms is illustrated in Figure~\ref{fig:diagram}. Most axioms are formulated for models involving a single feature, under which we establish that properties such as individual monotonicity and diminishing/increasing marginal effects are preserved. In models with multiple features, it is further possible to compare their relative importance within the same model. However, when considering different models, comparability is restricted to the relative importance of the same feature.

Beyond the current application, Shapley value–based explanations and interpretability techniques for LLMs can be utilized across diverse domains. For instance, extensive research has explored career path recommendations using LLMs \citep{jeon2025letters,wu2024exploring}, assisting individuals in selecting optimal next roles for career advancement. Shapley values help identify the key factors driving these recommendations. Another example is credit scoring, where narrative data can serve as features, particularly in low-income regions with limited financial records. In such cases, LLMs leverage these narratives to improve prediction accuracy \citep{teixeira2023enhancing,idbenjra2024investigating}, while Shapley values clarify the influence of specific narrative elements. Additional applications include fraud detection \citep{chang2024exposing}, algorithmic trading \citep{iacovides2024finllama}, and financial document consistency analysis \citep{wang2025assessing}, among others. Each of these benefits from Shapley values by enhancing trust and supporting better decision-making in complex AI systems. 

\begin{figure}[h]
    \includegraphics[scale=0.6]{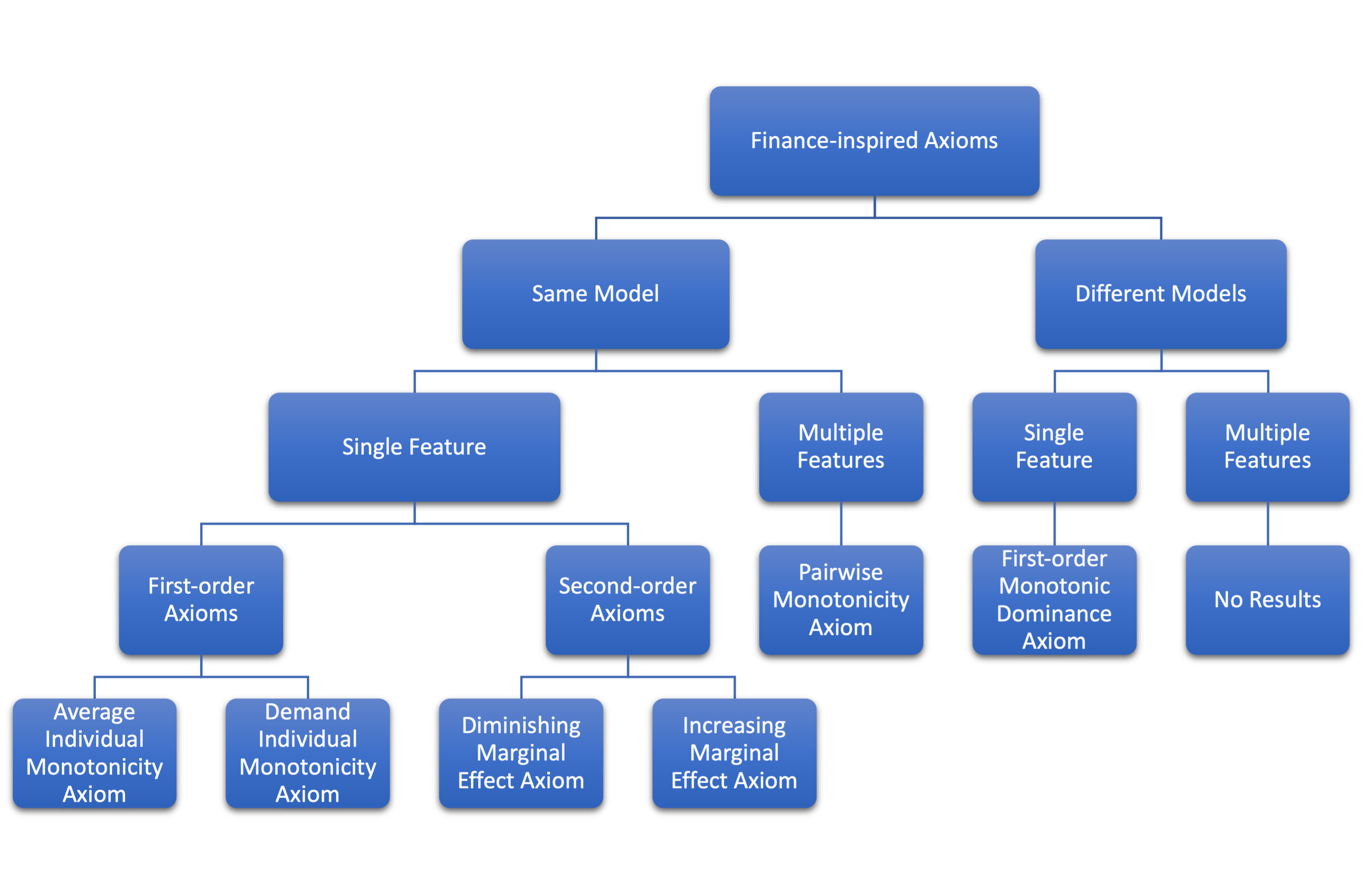}
    \caption{A diagram for finance-inspired axioms}
    \label{fig:diagram}
\end{figure}

\section{{Large-Scale Risk Attribution Based on Form 10-K}}

In Section~3, we analyzed the baseline Shapley value through simple yet intuitive illustrative examples. Building on these examples, we showed how the baseline Shapley value satisfies and preserves a collection of desirable domain knowledge-inspired axioms. 
In this section, we present a large-scale text-based experiment to demonstrate the practical applicability of the baseline Shapley value.

In recent years, textual disclosures have gained prominence as an alternative source of firm-level financial market information \citep{yang2015twitter,siano2025news}, complementing traditional return-based measures, and one direction is to quantify firm-level risk using textual data \citep{gupta2025using,zhu2025uncovering}. Conventional approaches quantify risk using the statistical properties of stock returns, such as volatility, beta, or downside risk, which are inherently backward-looking and reflect realized market outcomes. In contrast, textual risk measures extracted from corporate disclosures, particularly the risk factors section of regulatory filings, capture forward-looking assessments of potential adverse events as articulated by firm management. These narrative disclosures encode information about the nature, scope, and sources of uncertainty faced by the firm that may not yet be priced into returns or observable in historical data. As a result, text-based risk measures provide a structurally different view of risk, emphasizing anticipated vulnerabilities and exposure channels rather than ex-post market fluctuations, and have become an increasingly valuable tool for understanding firm risk in environments characterized by rapid change and evolving uncertainty.

\subsection{Form 10-K Item 1A — Risk Factors Section}

We focus on the Form 10-K in this study. The Form 10-K is the primary annual disclosure document that publicly traded firms are required to file with the U.S. Securities and Exchange Commission (SEC). It provides a comprehensive overview of a firm’s business operations, financial condition, and risk profile, including audited financial statements, management’s discussion and analysis, and detailed narrative disclosures. Unlike periodic market data, the form 10-K reflects management’s structured assessment of the firm’s performance and future uncertainties at a fixed reporting date, subject to regulatory standards and legal liability. As such, the 10-K serves as a central information source for investors, regulators, and researchers seeking to evaluate firm fundamentals and forward-looking risks.

Item 1A of the Form 10-K, the \textbf{risk factors section}, contains the firm’s formal disclosure of material risk factors that could adversely affect its business, financial condition, or future performance. This section is designed to provide forward-looking information by requiring management to identify and describe the most significant sources of uncertainty facing the firm, including operational, financial, regulatory, and macroeconomic risks. Unlike quantitative risk measures derived from historical outcomes, the risk factors section reflects management’s qualitative assessment of potential adverse scenarios and their transmission channels. Because these disclosures are subject to regulatory requirements and legal liability, the risk factors section represents a structured and economically meaningful narrative of firm-specific risk, making it a natural and increasingly important data source for textual analyses of corporate risk.

The risk factors section is typically organized hierarchically to present material risks in a structured and interpretable manner. In practice, the section may begin with an optional risk factors summary if the total disclosure exceeds fifteen pages, providing a concise overview of the most significant risks. The core of the risk factors section is organized under broad risk categories, each designated by a descriptive \textbf{risk heading} such as “Operational Risks,” “Financial Risks,” or “Legal Risks.” Within each heading, firms describe individual \textit{risk factors}, either as separate paragraphs or subheadings, explaining the nature of the risk, potential adverse consequences, and relevant transmission channels. Supporting details, including illustrative examples, regulatory context, or references to financial statements, may accompany each risk factor to enhance clarity.

LLMs provide a promising tool for systematically quantifying the qualitative information contained in the risk factors section. Unlike traditional approaches that rely on manual coding or dictionary-based text analysis, LLMs can process unstructured narrative text at scale, capture nuanced contextual meaning, and assess the relative severity or materiality of each disclosed risk. 
By inputting the text of risk factors along with an appropriate prompt, an LLM can generate risk scores that reflect both the inherent magnitude of the risk and its potential impact on firm performance. Such an approach enables researchers and practitioners to leverage textual disclosures as a forward-looking measure of uncertainty, complementing traditional return-based or accounting-based risk metrics, and can reveal risk exposures that may not yet be reflected in market prices or historical financial outcomes.

\subsection{Risk Attributions}

Risk attribution has long been a central topic in risk management. A substantial literature studies how individual risk factors contribute to overall portfolio risk; see, for example, \cite{meucci2005risk, rosen2010risk, shalit2021shapley, hagan2023portfolio}. Building on this line of work, \cite{tarashev2016risk} investigate risk attribution in the context of banking systems and employ it to construct measures of banks’ systemic importance. More recently, \cite{chen2025explaining} applied the Shapley value within a general framework for risk attribution and demonstrated its usefulness in applications such as factor modeling and option pricing. Overall, risk attribution provides a principled approach to decomposing aggregate risk into individual components, thereby enhancing our understanding of the sources of risk and supporting more effective risk management.

Unlike classical risk attribution problems, which typically take continuous variables or random variables as inputs, our setting uses textual data without preprocessing as the primary input. Risk attribution for textual data is inherently more challenging than in classical risk attribution settings. In traditional applications, additional structure, such as linear factor models, may be available to guide the attribution and provide economic insight. In contrast, when the inputs are textual data, the mapping from individual risk narratives to total risk is intrinsically nonlinear, context-dependent, and unstructured. As a result, classical linear attribution approaches are generally inadequate, necessitating more flexible attribution frameworks that can accommodate complex interactions among textual risk factors.
 Nevertheless, the Shapley value can be naturally extended to this context to provide meaningful risk attributions. 
 Such a decomposition is especially valuable for managers, regulators, and investors, as it facilitates the prioritization of attention, more efficient resource allocation, and targeted risk mitigation or hedging strategies based on the underlying sources of potential adverse outcomes, rather than relying solely on aggregated and retrospective risk measures.

 Given the scale and complexity of this setting, we focus primarily on the interpretation of the results, rather than on how the axioms are reflected in the analysis. Nevertheless, we emphasize that preserving the proposed domain knowledge-inspired axioms is essential for producing meaningful explanations. For example, because the content under each risk heading describes material risks, its attribution should always be nonnegative, as guaranteed by the average individual monotonicity axiom. An attribution method that assigns a negative contribution to a risk heading would yield an unreasonable explanation. Furthermore, if the content under a risk heading is revised to describe a more severe risk, its attribution should increase accordingly. This desirable property is captured by the demand individual monotonicity axiom. Overall, these domain knowledge-inspired axioms ensure that the explanations produced by the baseline Shapley value remain consistent with financial domain knowledge.

\subsection{Experiment Setup}

We now present a precise mathematical formulation of the problem and outline our proposed approach.
Mathematically, given the text $\overline{\*x}$ from the risk factors section, we consider the risk score by LLM as $f(\overline{\*x})$. We request the range of $f(\mathbf{x})$ within $[0,10]$, where a value of $0$ indicates no risk, and a value of $10$ indicates extremely high risk.
We treat all words within each risk heading as a single feature, as each heading conveys a clear financial meaning, making it readily interpretable for financial purposes. Accordingly, we partition the input as
\begin{align*}
    \*x = (\*x_1, \dots, \*x_n),
\end{align*}
where $\*x_i$ corresponds to the content of the $i$th risk heading. We then apply the baseline Shapley value to $f(\overline{\*x})$ to obtain the risk attribution $A_i$ associated with each heading. Specifically, we compute BShap over 30 independent runs and report the Monte Carlo mean together with the corresponding 95\% confidence interval. The Linearity axiom ensures that Shapley values are compatible with Monte Carlo estimation, as the Shapley value of the mean outcome equals the mean of the Shapley values, thereby providing theoretically consistent support for averaging attribution results across simulations.


\subsubsection{Discussion of Baselines and Prompt}

Here, we discuss the choice of baseline text. Two baselines are considered in the experiment:
\begin{enumerate}
    \item  A natural choice of baseline values in this setting is the empty input $\*x' = \{ \}$.
For long text inputs, using an empty input as the baseline is a natural and widely adopted choice. Long documents convey information gradually through sentences and paragraphs, so an empty input provides a clear and intuitive notion of “no information” against which the contribution of different parts of the text can be measured. The prompt corresponds to this choice of the baseline for the calculation of the characteristic function is provided in Appendix~\ref{sec:prompt1}.
    \item If the goal is to understand the dynamic changes in a firm's risk structure, then the data from a previous time point naturally serves as the baseline. That is, denoting the explicand at the year $t$ by $\overline{\*x}_t$, we may choose $\*x' = \overline{\*x}_s$ for some $s < t$.
     By comparing the explicand with this baseline, we effectively assess the relative changes in the content given the function. This approach is particularly natural for time-series data. The prompt for this choice of the baseline is provided in Appendix~\ref{sec:prompt2}. 
\end{enumerate}
Both baselines are valid, differing only in the purpose they serve: one emphasizes absolute attribution, while the other highlights relative attribution. Unlike the example in Section 3, where multiple word-level baseline choices were possible, for long textual data these two baselines represent the most natural and practical options.

\subsubsection{Discussion of Computational Cost}

Recent analysis\footnote{\url{https://corpgov.law.harvard.edu/2023/12/03/sec-risk-factors-disclosure-analysis/?utm_source=chatgpt.com}}
 of S$\&$P 500 firms’ Item 1A disclosures provide some summary statistics three years following the SEC’s 2020 reforms, which introduced requirements emphasizing materiality, clearer organization, and the reduction of boilerplate risk disclosures. On average, firms report five risk headings, with sections spanning approximately 13.5 pages. The most computationally expensive part of calculating the Shapley value is evaluating the characteristic function, which has a size of $2^n$, and each evaluation requires calling the LLM to compute the risk score. As a result, the overall computational cost of the Shapley value for an averaged file is approximately $\mathcal{O}(C \cdot 2^5 )$, which is computationally feasible. Here, $C$ denotes the cost of computing the risk score using the LLM, which depends both on the complexity of the model and on how the calculation and prompt are structured.

\subsubsection{LLM Model} \label{sec:GPT}

Quantifying financial risk from textual data is a challenging task. Since our primary objective is not to develop a specialized LLM for this purpose, we instead rely on an existing state-of-the-art GPT-based solver to address this problem. 
We interact with the GPT model through its application programming interface (API). Specifically, textual inputs, including task instructions, prompts, and optional system constraints, are sent to the model as structured requests. 
The API processes these inputs using a selected GPT model and returns generated outputs in a standardized response format. 
This API-based interaction enables scalable, automated inference, allowing the model to be seamlessly integrated into downstream pipelines for tasks such as risk scoring, text analysis, and decision support.
It should be noted that the results could be further improved through approaches such as fine-tuning the model. 
In addition, performing subsequent regression analyses on information extracted from Form~10-K filings could potentially enhance the robustness and accuracy of the estimated risk. 
Nevertheless, the primary focus of this experiment is to study the attributions of the estimated risk, rather than to optimize the accuracy of the risk measurement itself. Accordingly, we directly employ an existing GPT-based solver.

Specifically, the ChatGPT-5.2 API is used throughout the experiments. The ChatGPT-5.2 API from OpenAI is the latest generation of their LLMs, offering developers smarter, more contextually aware, and more flexible natural language understanding and generation capabilities compared to previous versions; it supports enhanced reasoning controls, longer context handling, and advanced instruction-following across use cases like conversational agents, content creation, complex data analysis, and automation, and includes optimized variants such as ChatGPT-5.2 and ChatGPT-5.2-Pro for deeper reasoning or professional workloads, all while maintaining compatibility with existing API workflows for interactive apps and batch tasks.

Experiments were performed on Google Colab using the default CPU runtime (2 vCPUs, approximately 12~GB RAM, Python~3.10), a widely used and readily accessible setup that allows most users to reproduce the results without specialized hardware or configuration. In particular, no GPUs, TPUs, or proprietary software dependencies are required to run the experiments.

\subsection{Examples}
We consider three examples here: Silicon Valley Bank, Wells Fargo \& Company, and Crocs, Inc. For the first two, we use empty baselines to examine their risk structures; for the third, we use the previous year's content as the baseline to analyze the relative changes in its risk structure. We discuss how to compare the resulting attributions under a unified risk heading in the end.

\subsubsection{An Example of Silicon Valley Bank}
As a demonstrating example, we look at the form 10-K of Silicon Valley Bank (SVB) in 2022 \footnote{\url{https://www.sec.gov/Archives/edgar/data/719739/000071973923000021/sivb-20221231.htm#ibb4dd73a1d3f4bff944b5d35fd2c5e2a_22}}. SVB was a prominent U.S.-based commercial bank that primarily served technology startups, venture capital firms, and high-growth companies. Established in 1983, SVB became a key financial partner for the innovation ecosystem, offering a range of services including lending, cash management, and deposit accounts tailored to the unique needs of technology-driven businesses. In March 2023, Silicon Valley Bank collapsed and was placed under regulatory receivership. The collapse was largely attributed to a combination of concentrated client exposure, a highly illiquid balance sheet, and a rapid increase in interest rates. In particular, the bank held a substantial portfolio of long-term fixed-income securities that declined sharply in value as interest rates rose, while simultaneously experiencing a sudden surge in deposit withdrawals from its concentrated startup client base. This asset--liability mismatch and lack of sufficient liquidity ultimately triggered a classic bank run, leading to the bank's failure. We select its Form~10\mbox{-}K filed in 2022, as the company collapsed in the following year, which allows us to better study the factors underlying its failure. Notably, interest rate risk and liquidity risk were emphasized in this filing, and we therefore examine whether the baseline Shapley value is able to capture the correspondingly large risk implied by these risk factors.

The risk factors section of SVB spans approximately 20 pages and contains six risk headings.
A summary of each risk heading can be found in Section~\ref{sec:SVB_risk_summary}. Note that the summary is provided by SVB and appears at the beginning of the section. More details follow the summary in the section. Interested readers are referred to SVB's original Form 10-K for the complete content.

We next apply the BShap to derive risk attributions for SVB. The Average computation time of the Shapley value took roughly 9.73 seconds, demonstrating that calculating the Shapley value for risk factor analysis on 10-K filings is computationally feasible. 

Overall, ChatGPT assigns an average score of 10.0 to SVB, indicating a very high level of risk reflected in SVB's 10-K filing. This assessment is consistent with the subsequent collapse of the SVB in the following year. The decomposition of the score using BShap, rounded to two decimal places is provided in Table~\ref{tab:risk_SVB}.
\begin{table}[htbp]
\centering
\caption{Risk attributions for Silicon Valley Bank in 2022 computed using BShap. We report the mean attributions along with their $95\%$ confidence intervals.
 }
\label{tab:risk_SVB}
\begin{tabular}{lcc}
\toprule
\textbf{Risk Heading} & \textbf{Mean} & \textbf{95$\%$ CI} \\
\midrule
Credit Risks & 4.04 & (3.92,4.15) \\
Market and Liquidity Risks & 2.75 & (2.65,2.85)\\
Operational Risks & 1.02 & (0.94,1.10) \\
Legal, Compliance and Regulatory Risks & 1.01 & (0.93,1.08) \\
Strategic, Reputational and Other Risks & 0.96 &  (0.89,1.03) \\
General Risk Factors & 0.22 & (0.20,0.25) \\
\bottomrule
\end{tabular}
\end{table}
Notably, the risks were dominated by \textit{Credit Risks} and \textit{Market and Liquidity Risks}. This is reasonable, as SVB clearly disclosed the significant pressure from credit and market/liquidity exposures. Credit risk is a significant risk because the bank’s loan portfolio is volatile and concentrated in borrowers with uncertain and cyclical financing needs, making credit losses difficult to predict and potentially exceeding existing loss reserves, especially during adverse economic conditions. In particular, SVB highlighted its potential exposure to rising interest rates and emphasized that liquidity risk could pose a major challenge in the \textit{Market and Liquidity Risk}.
As an illustrative example, the following risk factor from the \textit{Market and Liquidity Risk} section discusses this issue in greater detail:
\begin{quote}
    ``Our interest rate spread may further decline in the future. Any material reduction in our interest rate spread could have a material adverse effect on our business, results of operations, or financial condition.''
\end{quote}
Additional details in this risk factor provide further supporting evidence. For instance, it explicitly states:
\begin{quote}
    ``The Federal Reserve raised benchmark interest rates throughout 2022 and may continue to raise interest rates in response to economic conditions, particularly inflationary pressures. Continued increases in interest rates to combat inflation or otherwise may have unpredictable effects or minimize gains on our interest rate spread.''
\end{quote}
These factors were key contributors to SVB's collapse in the following year and were effectively captured by ChatGPT, with BShap appropriately attributing a large portion of the overall risk to them.  

The remaining risk categories are substantially smaller. A careful reading of the form 10-K shows that these risks are more general in nature and applicable to most firms. For instance, the \textit{General Risk Factors} heading begins with:
\begin{quote}
    ``If we fail to retain key employees or recruit new employees, or if we are unable to effectively manage the growth of our employee base, our growth and results of operations could be adversely affected.''
\end{quote}
While this represents a potential risk, it is broad and contains little specific information. Accordingly, such statements contribute minimally to the overall risk assessment. ChatGPT was able to recognize this, and baseline Shapley value correctly assigned a relatively small attribution.

\subsubsection{An Example of Wells Fargo $\&$ Company}

We consider another example with Wells Fargo \& Company (WFC) based on its 2024 10-K filing\footnote{\url{https://www.wellsfargo.com/assets/pdf/about/investor-relations/sec-filings/2024/exhibit-13.pdf}}. We focus on the 2024 filing, as the fourth-quarter in 2025 of the WFC's report is not yet available.
 Headquartered in San Francisco, California, WFC is a leading diversified financial services firm. Established in 1852, it offers a wide range of banking, investment, mortgage, and consumer and commercial finance services to clients across the United States and globally. The company operates through several business segments, including community banking, wholesale banking, and wealth and investment management, serving millions of individuals, small businesses, and large corporations. By assets, WFC ranks among the largest banks in the United States. We include WFC as an additional example alongside SVB, as the two banks, while operating within the same banking industry, exhibit distinct risk structures. This allows us to assess whether baseline Shapley value is able to capture and reflect these differences.

The Risk Factors section in WFC's 10-K spans roughly 14 pages and includes seven primary risk headings. 
A summary of risk factors are provided in Appendix~\ref{sec:WF_risk_summary}. Notably, the structure of the risk headings shows both overlaps and differences between WFC and SVB. WFC’s headings differ from those of SVB because its historical operational and regulatory challenges dominate its disclosures, whereas SVB emphasizes credit, market, and liquidity risks associated with its startup-focused balance sheet. This illustrates that risk headings reflect firm-specific constraints and experiences, rather than solely following regulatory taxonomy.

The average calculation of BShap took approximately 13.29 seconds. Overall, ChatGPT assigns WFC an average score of 10.0, indicating a very high level of risk. The decomposition of this risk score using BShap, rounded to two decimal place, is provided in Table~\ref{tab:risk_WFC}.
\begin{table}[htbp]
\centering
\caption{Risk attributions of the Wells Fargo $\&$ Company in 2024 computed using BShap. We report the mean attributions along with their $95\%$ confidence intervals. }
\label{tab:risk_WFC}
\begin{tabular}{lcc}
\toprule
\textbf{Risk Heading} & \textbf{Mean} & \textbf{95$\%$ CI} \\
\midrule
Economical, Financial Markets, Interest Rates, and Liquidity Risks & 2.32 & (2.23,2.40) \\
Regulatory Risks & 2.11 & (2.05,2.18) \\
Credit Risks & 2.16 & (2.11,2.21) \\
Operational, Strategic, and Legal Risks & 1.59 & (1.54,1.63) \\
Mortgage Business Risks & 1.16 & (1.11,1.20) \\
Competitive Risks & 0.36 & (0.33,0.39)\\
Financial Reporting Risks & 0.30 & (0.28,0.32) \\
\bottomrule
\end{tabular}
\end{table}
It is immediately noticeable that there is a significant difference in the structure of risk between WFC and SVB: WFC has a large exposure to Regulatory and Legal Risks, which exist in \textit{Regulatory Risks} and \textit{Operational, Strategic, and Legal Risks}. This exposure largely stems from major past incidents, most notably the 2016 fake-accounts scandal and the subsequent enforcement actions. These events exposed systemic weaknesses in governance, internal controls, and compliance processes, prompting the Federal Reserve, Consumer Financial Protection Bureau, and Office of the Comptroller of the Currency 
 to impose long-running consent orders, fines, and growth restrictions such as the asset cap. Because these regulatory actions directly constrain the bank’s ability to expand its balance sheet, deploy capital, and pursue business strategies, they represent persistent, binding risks. Consequently, WFC’s 10-K explicitly emphasizes regulatory and compliance risk, demonstrating how historical operational failures translate into ongoing vulnerability to regulatory scrutiny and potential penalties. Baseline Shapley value analysis is able to capture this firm-specific risk, highlighting its significance in WFC’s overall risk profile.

\subsubsection{An Example of Crocs, Inc}

We examine Crocs, Inc. for the years 2020 and 2021 to analyze the relative changes in its risk structure. We focus on 2020 because the company was significantly affected by the COVID-19 pandemic, allowing us to observe how its risk profile evolved in the subsequent year.
 Crocs, Inc. is a global footwear company best known for its lightweight, comfortable, and distinctive foam clogs. Founded in 2002 and headquartered in Niwot, Colorado, the company designs, develops, and markets casual footwear for men, women, and children under the Crocs brand, and has expanded its portfolio with acquisitions such as HEYDUDE in 2022. Crocs products are sold in more than 90 countries through a combination of retail stores, e-commerce platforms, and wholesale distributors, making it a prominent player in the casual footwear market.

To analyze the relative differences in the company over time in a time series context, a natural approach is to treat the content from 2020 as the baseline and the content from 2021 as the explicand.


The Risk Factors section have about 13 and 17 pages in 2020 and 2021, respectively. Six risk headings are considered.
A summary of the risk factors, with additional details, is provided in the Appendix~\ref{sec:prompt2}. Crocs clearly exhibits a different risk structure compared to WFC and SVB, reflecting the fact that they operate in entirely different industry sectors.

The average calculation of BShap took approximately 2 minutes and 8.6 seconds. 
This procedure required more time, as it employs a different prompt that invokes the API more frequently than the empty-baseline case.
Overall, ChatGPT assigns an average relative score of 5.33, indicating a moderate change of risk. The decomposition of this risk score using BShap, rounded to two decimal places, is provided in Table~\ref{tab:risk_CI}.
\begin{table}[htbp]
\centering
\caption{Risk attributions of the Crocs, Inc in 2021 computed using BShap. We report the mean attributions along with their $95\%$ confidence intervals. }
\label{tab:risk_CI}
\begin{tabular}{lcc}
\toprule
\textbf{Risk Heading} & \textbf{Mean} & \textbf{95$\%$ CI} \\
\midrule
Risks Related to Our Products & 0.11 & (0.01,0.21) \\
Risks Related to the Economy & -0.70 & (-0.85,-0.54) \\
Risks Related to Our Supply Chain & 1.04 & (0.94,1.14) \\
Risks Related to International Operations & 0.24 & (0.16,0.32) \\
Risks Specific to Our Company and Strategy & 1.84 & (1.77,1.92) \\
Financial and Accounting Risks & 2.80 & (2.73,2.87) \\
\bottomrule
\end{tabular}
\end{table}
First, we note that the largest changes in risk arise from the \textit{Risks Specific to Our Company and Strategy} and \textit{Financial and Accounting Risks} headings. Upon careful examination, the \textit{Risks Specific to Our Company and Strategy} heading includes the acquisition of HEYDUDE, which introduces additional content related to the acquisition. Consequently, there is a clear increase in the associated risk. Second, related to this acquisition, the resulting substantial indebtedness could adversely affect Crocs' business, financial condition, and operations, thereby increasing risk under the \textit{Financial and Accounting Risks} heading. Third, there is a decrease in risk under the \textit{Risks Related to the Economy} heading. This heading primarily focuses on COVID-19-related risks, which were mitigated in 2021 relative to 2020, although some impact remains. Overall, the BShap interpretation aligns well with the narrative behind the changes in the form.

\subsection{Check for Consistency and Fidelity} 

In this section, we evaluate baseline Shapley value using both consistency and fidelity. These two perspectives provide complementary ways to assess the validity of an explainable ML method.

\subsubsection{Consistency}

Consistency evaluates whether perturbations induce significant changes in the explanation \citep{alvarez2018towards}. Ensuring such consistency is crucial for producing meaningful explanations, particularly when inputs undergo perturbations. In the context of using an LLM API, a natural concern arises: since model outputs may vary across calls, how stable are the resulting explanations? To address this, we focus on \textit{rank consistency}. 

Specifically, let $\text{R}_i$ denote the rank of absolute value of the $i$th attribution in a given explanation, and let $\text{R}^*_i$ denote the rank when features are ranked based on their absolute value of averaged attributions. We then define the rank consistency measure as the mean absolute difference between these ranks:
\begin{align*}
    \text{Co}_i = \mathbb{E}\Big[ |\text{R}_i - \text{R}^*_i| \Big].
\end{align*}
Intuitively, a large value of $\text{Co}_i$ indicates that the rank is unstable across perturbations, whereas a small value indicates that the rank is consistent.

The results of three examples are reported in Table~\ref{tab:consistency}. Overall, we find that the rankings are highly consistent: when the mean attributions are well separated, the resulting feature ranks are very stable, whereas minor rank fluctuations occur when the estimated attributions are close in magnitude. The observed fluctuations are acceptable, as when attribution magnitudes are close, enforcing a strict relative ordering is not crucial. Our primary concern is ensuring that the relative ranking remains stable when the attributions differ substantially.

\begin{table}[htbp]
\centering
\caption{Consistency check of BShap-based risk attributions across SVB, WFC, and Crocs. The mean attribution is used here. }
\label{tab:consistency}
\begin{tabular}{lcc}
\toprule
\textbf{Risk Heading of SVB} & \textbf{ Attribution} & $\mathbf{\mathbb{E}\Big[ |\text{R}_i - \text{R}^*_i| \Big].}$ \\
\midrule
Credit Risks & 4.04 & 0.0 \\
Market and Liquidity Risks & 2.75 & 0.0 \\
Operational Risks & 1.02 & 0.97  \\
Legal, Compliance and Regulatory Risks & 1.01 & 0.67 \\
Strategic, Reputational and Other Risks & 0.96 &  1.03 \\
General Risk Factors & 0.22 & 0.0 \\
\midrule
\textbf{Risk Heading of WFC} & &  \\
\midrule
Economical, Financial Markets, & & \\
Interest Rates, and Liquidity Risks & 2.32 & 0.50 \\
Regulatory Risks & 2.11 & 0.77 \\
Credit Risks & 2.16 & 0.53 \\
Operational, Strategic, and Legal Risks & 1.59 & 0.00 \\
Mortgage Business Risks & 1.16 & 0.00 \\
Competitive Risks & 0.36 & 0.30\\
Financial Reporting Risks & 0.30 & 0.30 \\
\midrule
\textbf{Risk Heading of Crocs} & & \\
\midrule
Risks Related to Our Products & 0.11 & 0.40 \\
Risks Related to the Economy & -0.70 & 0.53 \\
Risks Related to Our Supply Chain & 1.04 & 0.13  \\
Risks Related to International Operations & 0.24 & 0.67  \\
Risks Specific to Our Company and Strategy & 1.84 & 0.07 \\
Financial and Accounting Risks & 2.80 & 0.00 \\
\bottomrule
\end{tabular}
\end{table}

\subsubsection{Fidelity}
Fidelity quantifies how faithfully an explanation captures the behavior of the underlying model. An explanation is said to have high fidelity if features identified as important exert a substantial influence on the model output, while features identified as unimportant have only a negligible effect. To evaluate fidelity, we conduct feature-removal tests by examining the change in the model prediction after removing either the most or the least important features, following the idea of \cite{samek2016evaluating,yuan2022explainability}. In this sense, fidelity assesses the truthfulness of an explanation, independent of interpretability or human plausibility.

Formally, given an index $k$, we consider the following two fidelity measures:
\begin{itemize}
    \item $\mathrm{Fi}_+(k)$ measures the effect of removing the $k$ most important features, defined as
    \[
        \mathrm{Fi}_+(k) = \mathbb{E}\Big[ \left| f(\overline{\*x}) - f(\*x'_{1:k}, \overline{\*x}_{k+1:n}) \right| \Big],
    \]
    where the expectation is taken over independent runs.
    \item $\mathrm{Fi}_-(k)$ measures the effect of removing the $k$ least important features, defined as
    \[
        \mathrm{Fi}_-(k) = \mathbb{E}\Big[ \left| f(\overline{\*x}) - f(\overline{\*x}_{1:n-k}, \*x'_{n+1-k:n}) \right| \Big].
    \]
\end{itemize}
Intuitively, we expect both scores to monotonically increase with $k$. Furthermore, we expect $\mathrm{Fi}_+(k) > \mathrm{Fi}_-(k)$, which indicates that the features identified as important indeed have a stronger influence on the model's output.

The result is provided in Figure~\ref{fig:fidelity}. Overall, we observe that removing the important risk factors has a substantially larger impact on the model output than removing nonimportant risk factors, which validates the attributions obtained using BShap.

\begin{figure}[htbp]
\centering
\begin{minipage}{0.32\textwidth}
    \centering
    \includegraphics[width=\textwidth]{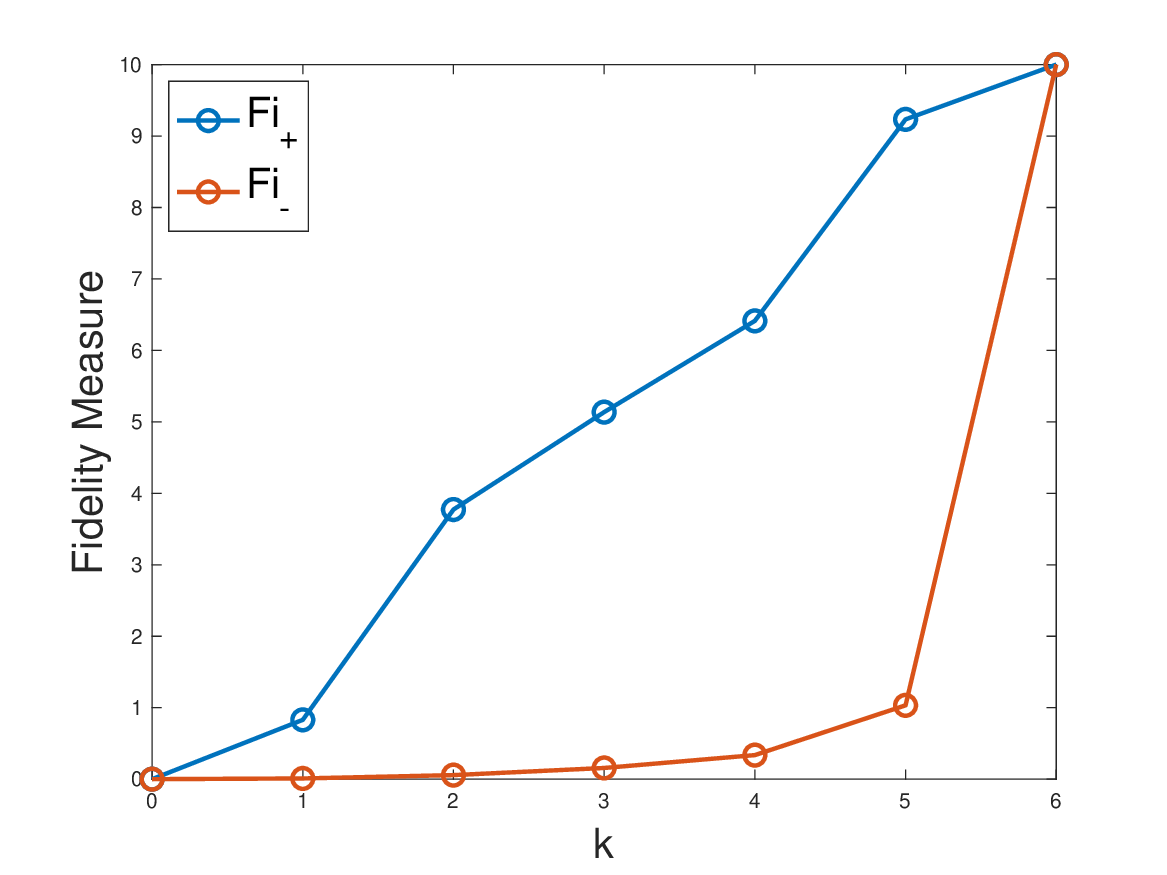}
    \caption{SVB}
\end{minipage}
\hfill
\begin{minipage}{0.32\textwidth}
    \centering
    \includegraphics[width=\textwidth]{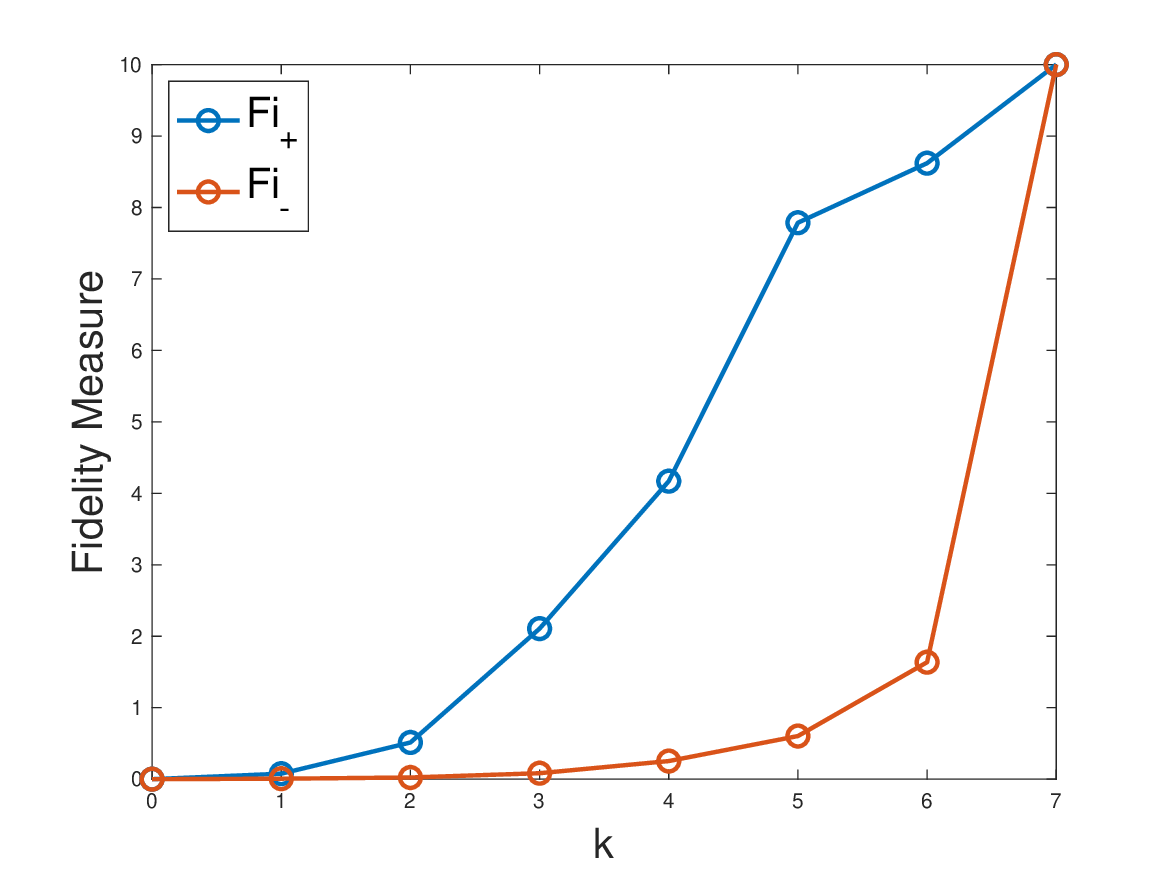}
    \caption{WFC}
\end{minipage}
\hfill
\begin{minipage}{0.32\textwidth}
    \centering
    \includegraphics[width=\textwidth]{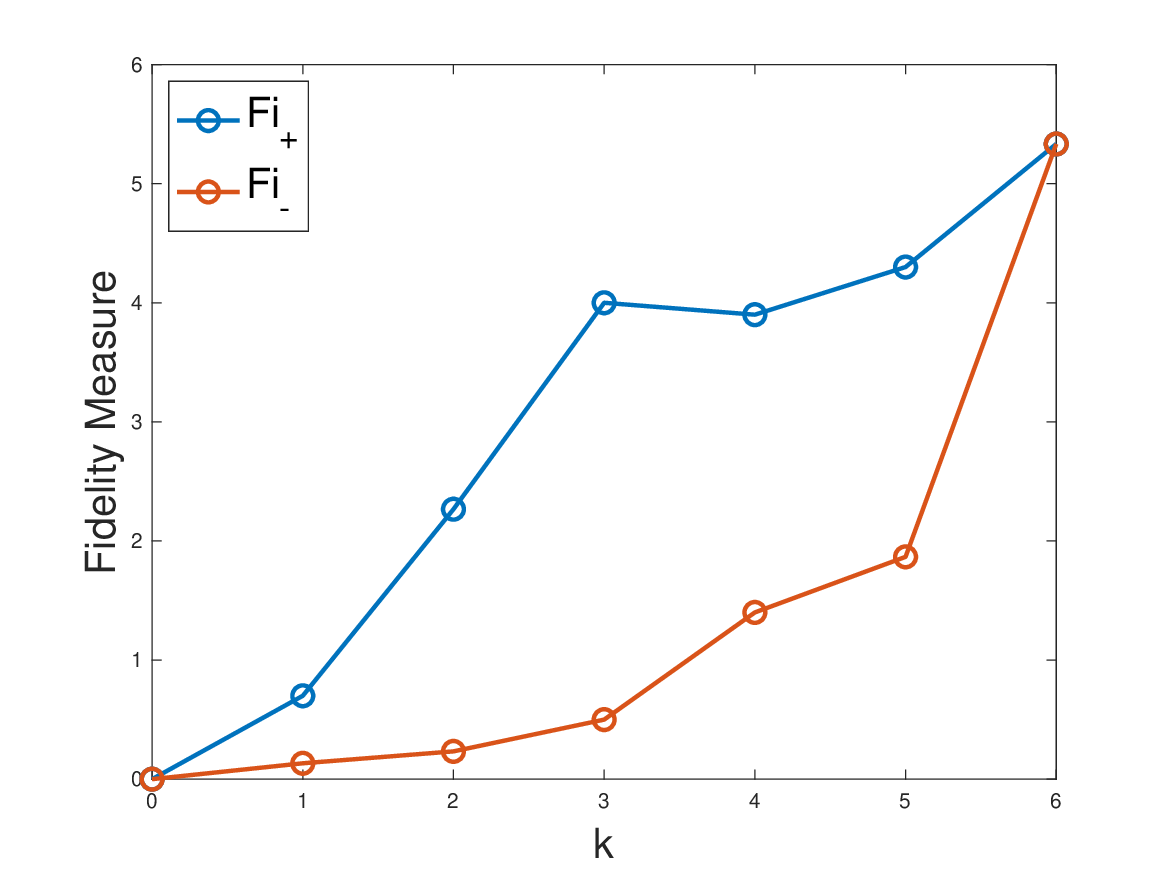}
    \caption{Crocs}
\end{minipage}
\caption{Fidelity check of BShap-based risk attributions across SVB, WFC, and Crocs.
 The x-axis record the number of removed features, whereas the y-axis calculate the fidelity score. }
\label{fig:fidelity}
\end{figure}

\subsection{{Using Unified Risk Headings}}

Since there is no universal standard for risk headings, firms typically adopt their own classifications to better reflect their specific risk structures. In certain settings, however, it is desirable to impose a common set of risk headings in order to facilitate comparison across firms, analogous to the use of common factors in factor models (e.g., industry factors). In this subsection, we describe how to implement such a unified risk classification.

First, it is worth noting that different firms may adopt substantially different risk headings due to the nature of their businesses. For example, the risk disclosures of Crocs, Inc.\ include a heading such as \textit{Risks Related to Our Products}, which would not typically appear in the context of banks. In such cases, there are no corresponding risk contents for banks, and we would expect the associated attribution to be zero, as the content under this heading is effectively empty.
Fortunately, this is guaranteed by the Dummy axiom satisfied by the Shapley value, as discussed in Theorem~\ref{thm:SH_uniq}. As a result, it is safe to directly assign zero attribution to such headings and omit the corresponding calculations in this scenario.

Next, using the examples of Wells Fargo \& Company and Silicon Valley Bank in the previous subsections, we illustrate how the Shapley value provides results under a unified set of risk headings. We consider the following commonly used risk headings for banks:
\begin{enumerate}
    \item Credit Risks
    \item Market, Interest Rate, and Liquidity Risks
    \item Operational and Technology Risks
    \item Regulatory and Legal Risks
    \item Strategic and External Risks
    \item Other Risks
\end{enumerate}
Note that the choice of unified risk headings is not unique. Alternative classifications may be adopted depending on the specific context or risk preferences; the above categorization is used here solely for illustrative purposes.

If a company's original risk heading matches or is contained within one of the unified risk headings, we directly assign the corresponding content to that unified category. Otherwise, we use ChatGPT to classify the risk factors into the appropriate unified risk headings. The prompt used for this classification is provided in Appendix~\ref{sec:prompt3}. The summarized risk factors under the unified risk headings for SVB and WFC are provided in Appendix~\ref{sec:SVB_risk_summary2} and Appendix~\ref{sec:WF_risk_summary2}, respectively.

We report the results using the unified risk headings for SVB and WFC in Table~\ref{tab:unified_risk}. Overall, compared with the original results in Table~\ref{tab:risk_SVB} and Table~\ref{tab:risk_WFC}, there are no substantial changes in the major risk headings, as the unified headings largely overlap with the original classifications.
We observe a noticeable increase in Credit Risk for WFC, driven by the reallocation of part of the original \textit{Mortgage Business Risk} into \textit{Credit Risk}. Comparing SVB and WFC directly, the most salient difference is that WFC exhibits significantly higher exposure to \textit{Regulatory and Legal Risk}, which is consistent with its 2016 fake-accounts scandal and subsequent enforcement actions.
In contrast, SVB shows higher levels of \textit{Credit Risk and Market, Interest Rate, and Liquidity Risk}, which is also reasonable given that, as a smaller, startup-oriented financial institution, it is more vulnerable to adverse macroeconomic conditions following the COVID-19 period.

\begin{table}[htbp]
\centering
\caption{Risk attributions of Silicon Valley Bank in 2022 and Wells Fargo \& Company in 2024 computed using BShap under the unified risk headings. We report the mean attributions along with their $95\%$ confidence intervals.}
\label{tab:unified_risk}
\begin{tabular}{lcc|cc}
\toprule
 & \multicolumn{2}{c|}{\textbf{SVB}} & \multicolumn{2}{c}{\textbf{WFC}} \\
\textbf{Risk Heading} & \textbf{Mean} & \textbf{95\% CI} & \textbf{Mean} & \textbf{95\% CI} \\
\midrule
Credit Risks & 4.10 & (3.97, 4.22) & 3.18 & (3.11, 3.25) \\
Market, Interest Rate, and Liquidity Risks & 2.96 & (2.88, 3.04) & 2.40 & (2.34, 2.45) \\
Operational and Technology Risks & 0.99 & (0.90, 1.09) & 1.56 & (1.50, 1.62)\\
Regulatory and Legal Risks & 1.10 & (1.02, 1.18) & 2.04 & (1.99, 2.09) \\
Strategic and External Risks & 0.64 & (0.59, 0.70) & 0.69 & (0.64, 0.74) \\
Other Risks & 0.21 & (0.19, 0.24) & 0.12 & (0.10, 0.13) \\
\bottomrule
\end{tabular}
\end{table}

\section{Conclusion}

This work investigates the performance of \textit{baseline Shapley values} in explaining predictions made by large language models  on textual data. To rigorously evaluate the effectiveness of attribution methods, we propose a set of axioms that emphasize \textit{fairness} and \textit{conceptual soundness}. These axioms are accompanied by formal mathematical proofs.
We show that baseline Shapley values satisfy key \textit{domain knowledge-inspired axioms}, enabling reasonable interpretation of the \textit{first and second-order main effects of a individual feature}, the \textit{relative importance of features within a single model}, and the \textit{comparative importance of the same feature across different models}. This axiomatic preservation makes the baseline Shapley value a reliable tool for interpreting feature attributions in text-based large language models.
In addition, we illustrate through concrete examples how baseline Shapley values can provide meaningful insights into the behavior of large language models.

{The use of the Shapley value can support risk management in finance in two respects. First, it facilitates the management of model risk. LLMs are inherently opaque, while financial applications require a clear understanding of model behavior. Our analysis further shows that the baseline Shapley value is consistent with financial intuition by preserving finance-inspired axioms, thereby enabling a principled interpretation of model outputs. 
Second, in many applications, the outputs of LLMs are directly related to risk. For example, in our risk attribution setting, the attributions reflect the risk associated with each risk heading relative to the total risk faced by the firm.
 More broadly, in applications such as credit scoring and fraud detection, LLMs can be used to quantify textual information describing potential customer risk. In these settings, the LLM output itself represents a risk measure, and the baseline Shapley value provides a decomposition that identifies the sources and relative contributions of these risks, thereby further supporting effective risk management.

 }

One notable limitation of the Shapley value is its \textit{computational complexity}. As the number of features increases, exact computation becomes intractable. A promising solution is the use of \textit{hierarchical explanations}, which aggregate features into structured groups.
Furthermore, while our current analysis centers on monotonicity-based principles, future work could explore \textit{richer structural patterns} in financial text data and develop new attribution frameworks tailored to the complexities of financial language and context.



\section*{Disclosure of Interest} The authors report that there are no competing interests to declare.

\section*{Data Availability Statement.} All data used in this paper are publicly available, and the relevant references are provided in the main text.

%

\bibliographystyle{tfcad}
\bibliography{sample-base}


\appendix

\section{Proofs}

\subsection{Proofs} \label{sec:proof}

\begin{proof}[Proof of Theorem~\ref{thm:AIM}]
    Given any subset $S$, we partition $\*x = (x_{\alpha},\*x_S,\*x_{N \backslash (S \cup \alpha)})$. 
    By assumption, we have
    \begin{align*}
    \text{BS}_{\alpha} = \sum_{S \subseteq N \backslash \alpha} \frac{|S|! (|N|-|S|-1)!}{|N|!} (f(\overline{x}_{\alpha},\overline{\*x}_{S},\*x'_{N \backslash (S \cup \alpha)}) - f(x_{\alpha}',\overline{\*x}_{S},\*x'_{N \backslash (S \cup \alpha)})) \geq 0.
\end{align*}
    
\end{proof}

\begin{proof}[Proof of Theorem~\ref{thm:DIM}]
Given any subset $S$, we partition $\*x = (x_{\alpha},\*x_S,\*x_{N \backslash (S \cup \alpha)})$. By assumption, we have
\begin{align*}
    & \text{BS}_{\alpha}((\overline{x}_{\alpha},\overline{\*x}_{\neg})) - \text{BS}_{\alpha}((x^*_{\alpha},\overline{\*x}_{\neg})) \\
    &= \sum_{S \subseteq N \backslash \alpha} \frac{|S|! (|N|-|S|-1)!}{|N|!} (f(\overline{x}_{\alpha},\overline{\*x}_{S},\*x'_{N \backslash (S \cup \alpha)}) - f(x^*_{\alpha},\overline{\*x}_{S},\*x'_{N \backslash (S \cup \alpha)})) \geq 0.
\end{align*}
\end{proof}

\begin{proof}[Proof of Theorem~\ref{thm:DME} and \ref{thm:IME}]
    We prove Theorem~\ref{thm:DME}, and the result for Theorem~\ref{thm:IME} is similar. 
    Given any subset $S$, we partition $\*x = (x_{\alpha},\*x_S,\*x_{N \backslash (S \cup \alpha)})$. First, $0 \leq \text{BS}_{\alpha}((\overline{x}_{\alpha},\overline{\*x}_{\neg})) - \text{BS}_{\alpha}((x^*_{\alpha},\overline{\*x}_{\neg}))$ is true by Theorem~\ref{thm:DIM}. For the second inequality, by discrete convexity, we have
    \begin{align*}
        & \text{BS}_{\alpha}((\overline{x}_{\alpha},\overline{\*x}_{\neg})) - \text{BS}_{\alpha}((x^*_{\alpha},\overline{\*x}_{\neg})) \\
         &= \sum_{S \subseteq N \backslash \alpha} \frac{|S|! (|N|-|S|-1)!}{|N|!} (f(\overline{x}_{\alpha},\overline{\*x}_{S},\*x'_{N \backslash (S \cup \alpha)}) - f(x^*_{\alpha},\overline{\*x}_{S},\*x'_{N \backslash (S \cup \alpha)})) \\
         & \leq \sum_{S \subseteq N \backslash \alpha} \frac{|S|! (|N|-|S|-1)!}{|N|!} (f(x_{\alpha}^*,\overline{\*x}_{S},\*x'_{N \backslash (S \cup \alpha)}) - f(x'_{\alpha},\overline{\*x}_{S},\*x'_{N \backslash (S \cup \alpha)})) \\
         &= \text{BS}_{\alpha}((x^*_{\alpha},\overline{\*x}_{\neg})).
    \end{align*}
\end{proof}

\begin{proof}[Proof of Theorem~\ref{thm:PM}]
    Recall that 
    \begin{align*}
        \text{BS}_{\alpha} = \sum_{S \subseteq N \backslash \alpha} \frac{|S|! (|N|-|S|-1)!}{N!} (f(\overline{x}_{\alpha},\overline{\*x}_{S},\*x'_{N \backslash (S \cup \alpha)}) - f(x_{\alpha}',\overline{\*x}_{S},\*x'_{N \backslash (S \cup \alpha)})).
    \end{align*}
    For $\text{BS}_{\beta}$, we utilize the symmetry. For each $S$ here, we consider $S'$ such that $x_{\alpha}$ and $x_{\beta}$ are swapped within $S$, and everything else is left unchanged. That is, if $\beta \notin S$, $S'=S$; if $\beta \in S$, then $S' = (S \backslash \beta) \cup \alpha$. For this arrangement, we have
    \begin{align*}
        & \text{BS}_{\alpha} - \text{BS}_{\beta} = \sum_{S \subseteq N \backslash \alpha} \frac{|S|! (|N|-|S|-1)!}{N!} \Delta f(S),
    \end{align*}
    where 
    \begin{align*}
        \Delta f(S) &= (f(\overline{x}_{\alpha},\overline{\*x}_{S},\*x'_{N \backslash (S \cup \alpha)}) - f(x_{\alpha}',\overline{\*x}_{S},\*x'_{N \backslash (S \cup \alpha)})) \\
        &- (f(\overline{x}_{\beta},\overline{\*x}_{S'},\*x'_{N \backslash (S' \cup \beta)}) - f(x_{\beta}',\overline{\*x}_{S'},\*x'_{N \backslash (S' \cup \beta})).
    \end{align*}
    Then, it is sufficient to show $\Delta f(S) \geq 0$. Given any subset $S$, we partition $\*x = (x_{\alpha},x_{\beta},\*x_{S\backslash \beta},\*x_{N \backslash (S \cup \alpha \cup \beta)})$.   
    If $\beta \notin S$, $S=S'$,  and we have    
    \begin{align*}
        \Delta f(S) &= (f(\overline{x}_{\alpha},x_{\beta}',\overline{\*x}_{S},\*x'_{N \backslash (S \cup \alpha \cup \beta)}) - f(x_{\alpha}',x_{\beta}',\overline{\*x}_{S},\*x'_{N \backslash (S \cup \alpha \cup \beta)})) \\
        &- (f(x_{\alpha}',\overline{x}_{\beta},\overline{\*x}_{S'},\*x'_{N \backslash (S' \cup \alpha \cup \beta)}) - f(x_{\alpha}',x_{\beta}',\overline{\*x}_{S'},\*x'_{N \backslash (S' \cup \alpha \cup \beta})) \\
        &= f(\overline{x}_{\alpha},x_{\beta}',\overline{\*x}_{S},\*x'_{N \backslash (S \cup \alpha \cup \beta)}) - f(x_{\alpha}',\overline{x}_{\beta},\overline{\*x}_{S'},\*x'_{N \backslash (S' \cup \alpha \cup \beta)}) \\
        &= f(\overline{x}_{\alpha},x_{\beta}',\overline{\*x}_{S},\*x'_{N \backslash (S \cup \alpha \cup \beta)}) - f(x_{\alpha}',\overline{x}_{\beta},\overline{\*x}_{S},\*x'_{N \backslash (S \cup \alpha \cup \beta)}).
    \end{align*}
    Therefore, $\Delta f(S) \geq 0$ by the assumption of pairwise monotonicity. 

    If $\beta \in S$, since $S \backslash \beta = S' \backslash \alpha$, $S \cup \alpha = S' \cup \beta$, then we have
    \begin{align*}
        \Delta f(S) &= (f(\overline{x}_{\alpha},\overline{x}_{\beta},\overline{\*x}_{S \backslash \beta},\*x'_{N \backslash (S \cup \alpha)}) - f(x_{\alpha}',\overline{x}_{\beta},\overline{\*x}_{S \backslash \beta},\*x'_{N \backslash (S \cup \alpha)})) \\
        &- (f(\overline{x}_{\alpha},\overline{x}_{\beta},\overline{\*x}_{S' \backslash \alpha},\*x'_{N \backslash (S'  \cup \beta)}) - f(\overline{x}_{\alpha},x_{\beta}',\overline{\*x}_{S' \backslash \alpha},\*x'_{N \backslash (S' \cup \beta})) \\
        &= f(\overline{x}_{\alpha},x_{\beta}',\overline{\*x}_{S' \backslash \alpha},\*x'_{N \backslash (S' \cup \beta)}) - f(x_{\alpha}',\overline{x}_{\beta},\overline{\*x}_{S \backslash \beta},\*x'_{N \backslash (S \cup \alpha)}) \\
        &=f(\overline{x}_{\alpha},x_{\beta}',\overline{\*x}_{S \backslash \beta},\*x'_{N \backslash (S \cup \alpha)}) - f(x_{\alpha}',\overline{x}_{\beta},\overline{\*x}_{S \backslash \beta},\*x'_{N \backslash (S \cup \alpha)}).
    \end{align*}
    Therefore, $\Delta f(S) \geq 0$ by the pairwise monotonicity assumption. Since $\Delta f(S)$ is always nonnegative, we conclude. 

\end{proof}

\begin{proof}[Proof of Theorem~\ref{thm:FMD}]
    Given any subset $S$, we partition $\*x = (x_{\alpha},\*x_S,\*x_{N \backslash (S \cup \alpha)})$. By linearity, we have
    \begin{align*}
    \text{BS}_{\alpha}(\overline{\*x},\*x',f) - \text{BS}_{\alpha}(\overline{\*x},\*x',g) &= \text{BS}_{\alpha}(\overline{\*x},\*x',f-g) \\
    &= \sum_{S \subseteq N \backslash \alpha} \frac{|S|! (|N|-|S|-1)!}{|N|!} \Delta f(S), 
\end{align*}
    whereas
    \begin{align*}
        \Delta f(S) &= ( f(\overline{x}_{\alpha},\overline{\*x}_S,\*x'_{N \backslash (S \cup \alpha)}) - f(x_{\alpha}',\overline{\*x}_S,\*x'_{N \backslash (S \cup \alpha)})) \\
        &- (g(\overline{x}_{\alpha},\overline{\*x}_S,\*x'_{N \backslash (S \cup \alpha)}) - g(x_{\alpha}',\overline{\*x}_S,\*x'_{N \backslash (S \cup \alpha)}))
    \end{align*}
    Therefore, $\Delta f(S) \geq 0$ by the assumption in the symmetric monotonicity axiom.  We conclude. 
\end{proof}

\section{{More Examples for Simple LLM Tasks}}

\subsection{A Comparison of GPT for Sector Classifications} \label{sec:sec_GPT}

For comparison, we use ChatGPT 5.2 under the same setup as described in Section~\ref{sec:GPT}. The following prompt is employed to compute the characteristic function:

\tcbset{colback=gray!10, colframe=black, left=4mm, right=4mm, top=2mm, bottom=2mm}
\begin{tcolorbox}

    You are an expert financial analyst. Your task is to classify companies into their appropriate industry sectors based on their descriptions. \\

    Instructions:
    \begin{itemize}
        \item Assign the company to one of the predefined sectors listed below.
        \item Return a list of probabilities for each sector, nothing else.
        \item Probabilities should sum to 1.
    \end{itemize}

    Predefined sectors:
    \begin{enumerate}
        \item Technology
        \item Finance
        \item Healthcare
        \item Consumer Goods
        \item Energy
        \item Others
    \end{enumerate}
    You have read the company description below. 

    Company description:
    \{$\overline{\*x}$\}

    Now, for each of the following $2^n$ = len(active mask) possible combinations
    of the description, compute the list of probabilities. 

    COALITIONS (each numbered):
    {coalition description}

    Return ONLY a list of probability list in order, one list per coalition. The order must match the numbering above.
    
\end{tcolorbox}

Since ChatGPT's outputs involve randomness, we run the experiment independently 30 times and report the averaged attributions. On average, GPT assigns $30\%$ probability to the Technology sector and $65\%$ probability to the Finance sector, compared to $64\%$ and $14\%$ assigned by the BLM model. Clearly, GPT provides more reasonable results that are better aligned with the true sector distribution.

Applying BShap to ChatGPT results yields the following attributions:
\begin{align*}
    \textbf{BS}_{\text{tech}} &= \begin{bmatrix} -0.12 & 0.29 & -0.02 \end{bmatrix}^T, \\
    \textbf{BS}_{\text{finance}} &= \begin{bmatrix} 0.36 & -0.01 & 0.10 \end{bmatrix}^T.
\end{align*}
We observe that the main difference between ChatGPT and the BLM model is that ChatGPT treats the phrase ``quantitative investment management" as a strong indicator of a financial company. Additionally, ChatGPT does not consider ``artificial intelligence and machine learning" as a negative indicator for a financial company, which is reasonable given that many financial firms heavily utilize AI in practice. These results suggest that ChatGPT outperforms the BLM model in this example, likely because it is more sensitive to financial terminology, a reasonable outcome given that ChatGPT represents the current state-of-the-art.

\section{{More Details in Risk Attribution Experiments}}

\subsection{Prompt} \label{sec:prompt}

\subsubsection{Prompt for the Characteristic Function with Empty Baselines} \label{sec:prompt1}
We use the following prompt to ask the ChatGPT API to calculate the risk score with the empty baseline.

\tcbset{colback=gray!10, colframe=black, left=4mm, right=4mm, top=2mm, bottom=2mm}
\begin{tcolorbox}
    You are a deterministic financial risk scoring function.

    You MUST follow these rules exactly:
    \begin{enumerate}
    \item Highly generic or boilerplate risk statements receive a score of 0.
    \item Vague language receives very low weight.
    \item Firm-specific, asymmetric, balance-sheet-relevant risks receive high scores.
    \item The empty coalition MUST receive 0.
    \item You MUST use the full 0–10 scale.
    \end{enumerate}

    You have read the FULL risk disclosure below:

    FULL DISCLOSURE:
    \{$\overline{x}$\}

    Now, for each of the following $2^n=$len(coalition risk headings) possible combinations
    of risk sections (coalitions), compute the risk score (0–10). The empty set
    should receive 0.

    COALITIONS (each numbered):
    \{coalition risk headings\}

    Return ONLY a list of numbers in order, one number per coalition.
    Example:
    [0.0, 1.2, 2.5, ..., 9.8]
    The order must match the numbering above.
\end{tcolorbox}

Within the prompt, the coalition of risk headings consists of combinations of individual risk headings used to calculate the characteristic functions for BShap. For example, for the first risk heading, we have:
\begin{quote}
    coalition risk headings = 1. Credit Risks
\end{quote}

Note that when the number of features becomes large, using the ChatGPT API to compute the characteristic function in a single prompt can produce outputs with incorrect dimensions, as the task may exceed the model's processing capacity. In such cases, we simply discard the affected calculations. This issue could be addressed by fine-tuning a local LLM or restoring certain model parameters; however, since model development is not the focus of this work, we do not pursue these solutions and instead concentrate on the model explanations.

\subsubsection{Prompt for the Characteristic Function with Baseline at Different Times} \label{sec:prompt2}

We discuss the prompt when employing the baseline for content from a previous time.
 For our previous approach that dealing with the empty baseline, we directly asked GPT to compute the risk score for the combination of risk headings after reading the content. However, due to limitations of the ChatGPT API, if we request the same computation for coalitions of differences, ChatGPT exhibits a trend effect, assigning the same sign to all coalitions once it determines which document is riskier. To address this, we calculate the relative risk independently by asking ChatGPT to recompute the relative risk score for each coalition, which increases computational time. Notably, this could be improved with a model that has local memory. Since the model itself is not the primary focus of this paper, we simply recompute each time and concentrate on BShap results.

Additionally, one could ask ChatGPT to evaluate the risk scores at different time separately and then apply the Shapley value to their difference. However, the current ChatGPT version is not trained specifically for this task and tends to assign large scores in all cases, causing the relative change to be easily overlooked. Therefore, it is preferable to directly request ChatGPT to provide the relative risk score between the two documents to emphasize the differences. As mentioned above, these limitations could be mitigated by improving the model, but doing so is beyond the scope of our study.

We provide the following prompt to ask the GPT API to calculate the risk score between $\*x'$ and $\overline{\*x}$. In our code, this is repeatedly applied to each coalitions to obtain the characteristic function. 
\tcbset{colback=gray!10, colframe=black, left=4mm, right=4mm, top=2mm, bottom=2mm}
\begin{tcolorbox}
    You are a deterministic financial risk scoring compare function.

    You have read the FULL risk disclosure for two documents below:

    DISCLOSURE A: \{"A. $\backslash$n $\backslash$n" + $\*x'$\}

    DISCLOSURE B: \{"B. $\backslash$n $\backslash$n" + $\overline{\*x}$\}

    Calculate the relative risk of DISCLOSURE B compared to DISCLOSURE A (B $-$ A).
    You MUST follow these rules exactly:
    \begin{enumerate}
        \item Return one numeric value on a scale of [-10, 10].
        \item Return 0 if the overall risk level is very close.
        \item Return a negative value if DISCLOSURE B is less risky than DISCLOSURE A.
        \item Return a positive value if DISCLOSURE B is more risky than DISCLOSURE A.
        \item Use magnitude to reflect strength:
        \begin{itemize}
        \item Strongly negative: $\leq$ -7 (DISCLOSURE B is much less risky)
        \item Moderately negative: -6 to -3
        \item Neutral: -2 to 2
        \item Moderately positive: 3 to 6
        \item Strongly positive: $\geq$ 7 (DISCLOSURE B is much more risky)
        \end{itemize}
    \end{enumerate}

    Return ONLY a number in [-10,10] scale.
\end{tcolorbox}

\subsection{{Prompt for Risk Factor Identification with Unified Headings}} \label{sec:prompt3}

We use the following prompt to determine the appropriate unified risk heading for each piece of content. As input, we consider only the risk factor titles (i.e., the subheadings of the detailed descriptions), which we find yield more accurate classifications with less noise when processed by ChatGPT. While this approach could be further improved through more refined training or model tuning, it is sufficient for our purposes and not the primary focus of this study.

\tcbset{colback=gray!10, colframe=black, left=4mm, right=4mm, top=2mm, bottom=2mm}
\begin{tcolorbox}
    You are given a list of paragraphs from a company filing: \{full text\}.

    Your task is to classify each paragraph into one of the following risk headings:\{headings\}.

    Return the result as a JSON list of lists. Each sublist corresponds to one risk heading in the same order as headings. 
    Each sublist contains the indices of paragraphs that belong to that risk heading.
    The number of elements in the output list must equal the number of input paragraphs.
\end{tcolorbox}

\subsection{Summary of Risk Factors} 
In this section, we provide a summary of the risk factors. Broadly speaking, the summary includes the risk factors listed under each risk heading, while detailed discussions of individual risk factors are omitted. Interested readers are referred to the original Form~10-K for further details.

\subsubsection{Summary of SVB's Risk Factors}\label{sec:SVB_risk_summary}

\begin{enumerate}
\item Credit Risks
\begin{enumerate}
\item Because of the credit profile of our loan portfolio, our levels of nonperforming assets and charge-offs can be volatile, and we may need to make material provisions for credit losses in any period.
\item Our ACL is determined based upon both objective and subjective factors, and may not be adequate to absorb any actual credit losses.
\item The borrowing needs of our clients have been and may continue to be unpredictable, especially during a challenging economic environment. We may not be able to meet our unfunded credit commitments, or adequately reserve for losses, which could have a material adverse effect.
\end{enumerate}
\item Market and Liquidity Risks
\begin{enumerate}
\item Instability and adverse developments in national or global financial markets and overall economic conditions, including as a result of geopolitical matters, may materially affect our business, financial condition and results of operations.
\item Our interest rate spread may decline further in the future. Any material reduction in our interest rate spread could have a material adverse effect on our business, results of operations or financial condition.
\item Liquidity risk could impair our ability to fund operations and jeopardize our financial condition.
\item Our equity warrant assets, venture capital and private equity fund investments and direct equity investment portfolio gains and losses depend upon the performance of our portfolio investments and the general condition of the public and private equity and M$\&$A markets which are uncertain and may vary materially by period.
\item Changes in the market for public equity offerings, M$\&$A or a slowdown in private equity or venture capital investment levels have affected and may continue to affect the needs of our clients for investment banking or M$\&$A advisory services and lending products, which could adversely affect our business, results of operations or financial condition.
\end{enumerate}
\item Operational Risks
\begin{enumerate}
\item Fraudulent activity could have a material adverse effect on our business, financial condition or results of operations.
\item A data breach, disruption of service or other cybersecurity-related incident could have a material adverse effect on our business, financial condition or results of operations.
\item We face risks associated with the ability of our IT systems and our people and processes to support our operations and future growth effectively.
\item Business disruptions due to natural disasters and other external events beyond our control, including pandemics, have in the past adversely affected our business, financial condition or results of operations and may do so in the future.
\item We face risks from a prolonged work-from-home arrangement as well as our implementation of a broader plan to return to the office.
\item We face risks from our reliance on business partners, service providers and other third parties.
\item The soundness of other financial institutions could adversely affect us.
\item We depend on the accuracy and completeness of information about customers and counterparties.
\item We face risks associated with our current international operations and ongoing international expansion.
\item Our holding company, SVB Financial, relies on equity warrant assets income, investment distributions and dividends from its subsidiaries for most of its cash revenues.
\item Climate change has the potential to disrupt our business and adversely impact the operations and creditworthiness of our clients.
\item The COVID-19 pandemic created significant economic and financial disruptions that adversely affected certain aspects of our business and operations and such disruptions have the potential to reoccur.
\end{enumerate}
\item Legal, Compliance and Regulatory Risks
\begin{enumerate}
\item We are subject to extensive regulation that could limit or restrict our activities, impose financial requirements or limitations on the conduct of our business, or result in higher costs to us, and the stringency of the regulatory framework applicable to us may increase if, and as, our balance sheet continues to grow.
\item As a bank holding company with more than $\$$100 billion of average total consolidated assets, we are subject to stringent regulations, including certain enhanced prudential standards applicable to large bank holding companies. If we exceed certain other thresholds, we will become subject to even more stringent regulations.
\item We face a risk of noncompliance and enforcement action with the Bank Secrecy Act, other anti-money laundering and anti-bribery statutes and regulations and U.S. economic and trade sanctions.
\item If we were to violate, or fail to comply with, international, federal or state laws or regulations governing financial institutions, we could be subject to disciplinary action or litigation that could have a material adverse effect on our business, financial condition, results of operations or reputation.
\item Laws and regulations regarding the handling of personal data and information may impede our services or result in increased costs, legal claims or fines against us.
\item Adverse results from litigation or governmental or regulatory investigations can impact our business practices and operating results.
\item A failure to identify and address potential conflicts of interest could adversely affect our businesses.
\item Anti-takeover provisions and federal laws may prevent a merger or acquisition that may be attractive to stockholders and/or have an adverse effect on our stock price.
\end{enumerate}
\item Strategic, Reputational and Other Risks
\begin{enumerate}
\item We have experienced significant growth during 2021 and into 2022, including deposit growth. If we again experience deposit growth at a similar or greater rate than has occurred in the past, we may need to raise additional equity to support our capital ratios.
\item Concentration of risk increases the potential for significant losses, while the establishment of limits to mitigate concentration risk increases the potential for lower revenues and slower growth.
\item Decreases in the amount of equity capital available to our clients could adversely affect us.
\item We face competitive pressures that could adversely affect our business, financial results, or growth.
\item Our ability to maintain or increase our market share depends on our ability to attract and maintain, as well as meet the needs of, existing and future clients.
\item We face risks in connection with our strategic undertakings and new business initiatives.
\item We may fail to realize growth prospects and benefits anticipated as a result of the Boston Private acquisition.
\item Any damage to our reputation and relationships could have a material adverse effect on our business.
\item An ineffective risk management framework could have a material adverse effect on our strategic planning and our ability to mitigate risks and/or losses and could have adverse regulatory consequences.
\item We do not currently pay dividends on shares of our common stock and may not do so in the future.
\end{enumerate}
\item General Risk Factors:
\begin{enumerate}
    \item If we fail to retain key employees or recruit new employees, or if we are unable to effectively manage the growth of our employee base, our growth and results of operations could be adversely affected.
    \item If we fail to maintain an effective system of internal control over financial reporting, we may not be able to accurately report our financial results. As a result, current and potential holders of our securities could lose confidence in our financial reporting, which would harm our business and the trading price of our securities.
    \item Changes in accounting standards could materially impact our financial statements.
    \item We could be adversely affected by changes in tax laws and regulations or their interpretations.
    \item We rely on quantitative and qualitative models to measure risks and to estimate certain financial values.
    \item The price of our capital stock may be volatile or may decline.
    \item Our capital stock is subordinate to our existing and future indebtedness.
\end{enumerate}
\end{enumerate}

\subsubsection{Summary of WFC's Risk Factors} \label{sec:WF_risk_summary}

\begin{enumerate}
    \item Economical, Financial Markets, Interest Rates, and Liquidity Risks
    \begin{enumerate}
        \item Our financial results have been, and will continue to be, materially affected by general economic conditions, and a deterioration in economic conditions or in the financial markets may materially adversely affect our lending and other
        businesses and our financial results and condition.
        \item Changes in interest rates and financial market values could
        reduce our net interest income and earnings, as well as our
        other comprehensive income, including as a result of
        recognizing losses on the debt and equity securities that we
        hold in our portfolio or trade for our customers.
        \item Effective liquidity management is essential for the operation
        of our business, and our financial results and condition could be
        materially adversely affected if we do not effectively manage
        our liquidity.
        \item Adverse changes in our credit ratings could have a material
        adverse effect on our liquidity, cash flows, and financial results
        and condition.
        \item We rely on dividends from our subsidiaries for liquidity, and
        federal and state law, regulatory requirements, and certain
        contractual arrangements can limit those dividends.
    \end{enumerate}
    \item Regulatory Risks 
    \begin{enumerate}
        \item Current and future legislation and/or regulation could require
        us to change certain of our business practices, reduce our
        revenue and earnings, impose additional costs on us or
        otherwise adversely affect our business operations and/or
        competitive position.
        \item We could be subject to more stringent capital, leverage or
        liquidity requirements or restrictions on our growth, activities
        or operations if regulators determine that our resolution or
        recovery plan is deficient.
        \item Our security holders may suffer losses in a resolution of
        Wells Fargo even if creditors of our subsidiaries are paid in full.
        \item Regulatory rules and requirements may impose higher capital
        and liquidity levels, limiting our ability to pay common stock
        dividends, repurchase our common stock, invest in our
        business, or provide loans or other products and services to our
        customers.
        \item FRB policies, including policies on interest rates, can
        significantly affect business and economic conditions and our
        financial results and condition.
    \end{enumerate}
    \item Credit Risks
    \begin{enumerate}
        \item Increased credit risk, including as a result of a deterioration in
        economic conditions or changes in market conditions, could
        require us to increase our provision for credit losses and
        allowance for credit losses and could have a material adverse
        effect on our results of operations and financial condition.
        \item We may have more credit risk and higher credit losses to the
        extent our loans are concentrated by loan type, industry
        segment, borrower type, or location of the borrower or
        collateral.
    \end{enumerate}
    \item Operational, Strategic, and Legal Risks
    \begin{enumerate}
        \item A failure in or breach of our operational or security systems,
        controls or infrastructure, or those of our third-party vendors
        and other service providers, could disrupt our businesses,
        damage our reputation, increase our costs and cause losses.
        \item A cyberattack or other information security incident could
        have a material adverse effect on our results of operations,
        financial condition, or reputation.
        \item Our framework for managing risks may not be fully effective in
        mitigating risk and loss to us.
        \item We may be exposed to additional legal or regulatory
        proceedings, costs, and other adverse consequences related to
        instances where customers may have experienced financial
        harm.
        \item We may incur fines, penalties, business restrictions, and other
        adverse consequences from regulatory violations or from any
        failure to meet regulatory standards or expectations.
        \item Reputational harm, including as a result of our actual or alleged
        conduct or public opinion of the financial services industry
        generally, could adversely affect our business, results of
        operations, and financial condition.
        \item If we are unable to develop and execute effective business
        plans or strategies or manage change effectively, our
        competitive standing and results of operations could suffer.
        \item Our operations and business could be adversely affected by the
        impacts of climate change.
        \item We are exposed to potential financial loss or other adverse
        consequences from legal actions.
    \end{enumerate} 
    \item Mortgage Business Risks 
    \begin{enumerate}
        \item Our mortgage banking revenue can be volatile from quarter to
        quarter, including from the impact of changes in interest rates,
        and we rely on the GSEs to purchase our conforming loans to
        reduce our credit risk and provide liquidity to fund new
        mortgage loans.
        \item We may suffer losses, penalties, or other adverse
        consequences if we fail to satisfy our obligations with respect
        to the residential mortgage loans or other assets we originate
        or service.
    \end{enumerate}
    \item Competitive Risks
    \begin{enumerate}
        \item We face significant and increasing competition in the rapidly
        evolving financial services industry.
        \item Our ability to attract and retain qualified employees is critical
        to the success of our business and failure to do so could adversely affect our business performance, competitive
        position and future prospects.
    \end{enumerate}
    \item Financial Reporting Risks
    \begin{enumerate}
        \item Changes in accounting standards, and changes in how
        accounting standards are interpreted or applied, could
        materially affect our financial results and condition.
        \item Our financial statements require certain assumptions,
        judgments, and estimates and rely on the effectiveness of our
        internal control over financial reporting.
    \end{enumerate} 
\end{enumerate}

\subsubsection{Summary of Crocs's Risk Factors} \label{sec:Croc_risk_summary}

First, we discuss some data processing steps.
\begin{itemize}
    \item Crocs has a different number of risk headings across the two years; specifically, the 2021 Form~10-K includes an additional heading, \emph{HEYDUDE Acquisition Risks}. To facilitate a more direct and interpretable comparison, we merge the content under this heading into \emph{Risks Specific to Our Company and Strategy}. This adjustment is justified for two reasons: first, the acquisition naturally fits within this category; second, if left separate, the new heading would receive a large positive relative risk attribution, since it was previously empty, yielding a trivial BShap result. 
    \item Some risk factors shifted headings between the two years. To ensure consistency, we assign such risk factors to the same heading as in 2020.
    \item There is substantial overlap between the two forms across consecutive years. To focus on the relative differences and reduce computational complexity, we remove risk factors that are identical in both years and retain only those that differ.
\end{itemize}

Below, we provide a summary of the risk factors section. Note that Crocs does not explicitly provide a summary; instead, we generate the summary following the style used for the SVB.

Summary for 2020: 
\begin{enumerate}
    \item Risks Related to Our Products
    \begin{enumerate}
        \item Our success depends substantially on the value of our brand; failure to strengthen and preserve this value, either through our actions or those of our business partners, could have a negative impact on our financial results.
        \item We face significant competition.
        \item Continuing to rationalize our existing product assortment and introducing new products may be difficult and expensive. If we are unable to do so successfully, our brand may be adversely affected and we may not be able to maintain or grow our current revenue and profit levels.
        \item Failure to adequately protect our trademarks and other intellectual property rights and counterfeiting of our brand could divert sales, damage our brand image and adversely affect our business.
        \item Failure to continue to obtain or maintain high-quality endorsers of our products could harm our business.
        \item We rely on technical innovation to compete in the market for our products.
    \end{enumerate}
    \item Risks Related to the Economy
    \begin{enumerate}
        \item The COVID-19 pandemic has had an adverse impact, and may have a future material adverse impact, on our business, operations, liquidity, financial condition, and results of operations.
    \end{enumerate} 
    \item Risks Related to Our Supply Chain
    \begin{enumerate}
        \item If we do not accurately forecast consumer demand, we may have excess inventory to liquidate or have greater difficulty filling our customers’ orders, either of which could adversely affect our business.
        \item Our third-party manufacturing operations must comply with labor, trade, and other laws. Failure to do so may adversely affect us.
        \item We depend solely on third-party manufacturers located outside of the U.S.
        \item We depend on a number of suppliers for key production materials, and any disruption in the supply of such materials could interrupt product manufacturing and increase product costs.
    \end{enumerate} 
    \item Risks Related to International Operations
    \begin{enumerate}
        \item Changes in foreign exchange rates, most significantly but not limited to the Euro, Russian Ruble, Brazilian Real, South Korean Won, Japanese Yen, or other global currencies could have a material adverse effect on our business and financial results.
        \item We conduct significant business activity outside the U.S., which exposes us to risks of international commerce.
        \item Changes in global economic conditions may adversely affect consumer spending and the financial health of our customers and others with whom we do business, which may adversely affect our financial condition, results of operations, and cash resources.
    \end{enumerate}
    \item Risks Specific to Our Company and Strategy
    \begin{enumerate}
        \item We may be unable to successfully execute our long-term growth strategy, maintain or grow our current revenue and profit levels, or accurately forecast demand and supply for our products.
        \item Our business relies significantly on the use of information technology. A significant disruption to our operational technology or those of our business partners, a privacy law violation, or a data security breach could harm our reputation and/or our ability to effectively operate our business, and our financial results.
        \item If our online e-commerce sites, or those of our customers, do not function effectively, our business and financial results could be materially adversely affected.
        \item Our financial success depends in part on the strength of our relationships with, and the success of, our wholesale and distributor customers.
        \item Operating company-operated retail stores incurs substantial fixed costs. If we are unable to generate sales, operate our retail stores profitably, or otherwise fail to meet expectations, we may be unable to reduce such fixed costs and avoid losses or negative cash flows.
        \item Our financial results may be adversely affected if substantial investments in businesses and operations fail to produce expected returns.
        \item We depend on employees across the globe, the loss of whom would harm our business.
        \item We are subject to periodic litigation, which could result in unexpected expenditures of time and resources.
        \item Extreme weather conditions, natural disasters, public health issues, or other events outside of our control could negatively impact our operating results and financial condition.
        \item Our restated certificate of incorporation, amended and restated bylaws, and Delaware law contain provisions that could discourage a third party from acquiring us and consequently decrease the market value of an investment in our stock.
    \end{enumerate} 
    \item Financial and Accounting Risks
    \begin{enumerate}
        \item We may be required to record impairments of long-lived assets or incur other charges relating to our company-operated retail operations.
        \item Our quarterly revenues and operating results are subject to fluctuation as a result of a variety of factors, which could increase the volatility of the price of our common stock.
        \item Our senior revolving credit facility agreement (as amended to date, the “Credit Agreement”) contains financial covenants that require us to maintain certain financial measures and ratios and includes restrictive covenants that limit our ability to take certain actions. A breach of any of those restrictive covenants may cause us to be in default under the Credit Agreement, and our lenders could foreclose on our assets.
        \item Changes in the method for determining LIBOR and/or the potential replacement of LIBOR could adversely affect our results of operations.
        \item The risks of maintaining significant cash abroad could adversely affect our cash flows in the U.S., our business, and financial results.
        \item Changes in tax laws and unanticipated tax liabilities and adverse outcomes from tax audits or tax litigation could adversely affect our effective income tax rate and profitability.
        \item We may fail to meet analyst and investor expectations, which could cause the price of our stock to decline.
    \end{enumerate} 
\end{enumerate}

Overall, the structure of the risk factors section in 2021 is very similar to that of 2020, although the details within each risk factor exhibit some differences. Specifically, the 2021 Form~10-K includes additional risk factors compared to 2020, some contents in the risk factors were modified, but no risk factors were removed. For illustration, we provide below a summary of the risk headings, highlighting the new risk factors introduced in 2021.

\begin{enumerate}
    \item Risks Related to Our Supply Chain
    \begin{enumerate}
        \item Supply chain disruptions could interrupt product manufacturing and global logistics and increase product costs.
        \item Our operations are dependent on the global supply chain and impacts of supply chain constraints and inflationary pressure could adversely impact our operating results.
        \item We depend on a number of suppliers for key production materials, and any disruption in the supply of such materials could interrupt product manufacturing and increase product costs.
    \end{enumerate} 
    \item Risks Specific to Our Company and Strategy
    \begin{enumerate}
        \item Our ability to realize the benefits from the HEYDUDE Acquisition is substantially dependent on our ability to continue to grow HEYDUDE.
        \item The announcement and pendency of the HEYDUDE Acquisition may have an adverse effect on our business, financial condition, operating results, and cash flows.
        \item The incurrence by us of substantial indebtedness in connection with the financing of the HEYDUDE Acquisition may have an adverse impact on our liquidity, limit our flexibility in responding to other business opportunities, and increase our vulnerability to adverse economic and industry conditions.
    \end{enumerate} 
    \item Financial and Accounting Risks
    \begin{enumerate}
        \item Our indebtedness could adversely affect our business, financial condition, and results of operations, as well as the ability to meet payment obligations under our Revolving Credit Agreement and the Notes.
        \item Despite our current level of indebtedness, we may be able to incur substantially more debt, which could increase the risks to our financial condition described above.
    \end{enumerate} 
\end{enumerate}

\subsubsection{Summary of SVB's Risk Factors under the Unified Risk Headings}\label{sec:SVB_risk_summary2}
\begin{enumerate}
    \item Credit Risks
    \begin{enumerate}
        \item Because of the credit profile of our loan portfolio, our levels of nonperforming assets and charge-offs can be volatile, and we may need to make material provisions for credit losses in any period.
        \item Our ACL is determined based upon both objective and subjective factors, and may not be adequate to absorb any actual credit losses.
        \item The borrowing needs of our clients have been and may continue to be unpredictable, especially during a challenging economic environment. We may not be able to meet our unfunded credit commitments, or adequately reserve for losses, which could have a material adverse effect.
        \item Decreases in the amount of equity capital available to our clients could adversely affect us.
    \end{enumerate}
    \item Market, Interest Rate, and Liquidity Risks
    \begin{enumerate}
        \item Instability and adverse developments in national or global financial markets and overall economic conditions, including as a result of geopolitical matters, may materially affect our business, financial condition and results of operations.
        \item Our interest rate spread may decline further in the future. Any material reduction in our interest rate spread could have a material adverse effect on our business, results of operations or financial condition.
        \item Liquidity risk could impair our ability to fund operations and jeopardize our financial condition.
        \item Our equity warrant assets, venture capital and private equity fund investments and direct equity investment portfolio gains and losses depend upon the performance of our portfolio investments and the general condition of the public and private equity and M$\&$A markets which are uncertain and may vary materially by period.
        \item Changes in the market for public equity offerings, M$\&$A or a slowdown in private equity or venture capital investment levels have affected and may continue to affect the needs of our clients for investment banking or M$\&$A advisory services and lending products, which could adversely affect our business, results of operations or financial condition.
        \item Concentration of risk increases the potential for significant losses, while the establishment of limits to mitigate concentration risk increases the potential for lower revenues and slower growth.
        \item Decreases in the amount of equity capital available to our clients could adversely affect us.
    \end{enumerate}
    \item Operational and Technology Risks
    \begin{enumerate}
        \item Fraudulent activity could have a material adverse effect on our business, financial condition or results of operations.
        \item A data breach, disruption of service or other cybersecurity-related incident could have a material adverse effect on our business, financial condition or results of operations.
        \item We face risks associated with the ability of our IT systems and our people and processes to support our operations and future growth effectively.
        \item Business disruptions due to natural disasters and other external events beyond our control, including pandemics, have in the past adversely affected our business, financial condition or results of operations and may do so in the future.
        \item We face risks from a prolonged work-from-home arrangement as well as our implementation of a broader plan to return to the office.
        \item We face risks from our reliance on business partners, service providers and other third parties.
        \item The soundness of other financial institutions could adversely affect us.
        \item We depend on the accuracy and completeness of information about customers and counterparties.
        \item We face risks associated with our current international operations and ongoing international expansion.
        \item Our holding company, SVB Financial, relies on equity warrant assets income, investment distributions and dividends from its subsidiaries for most of its cash revenues.
        \item Climate change has the potential to disrupt our business and adversely impact the operations and creditworthiness of our clients.
        \item The COVID-19 pandemic created significant economic and financial disruptions that adversely affected certain aspects of our business and operations and such disruptions have the potential to reoccur.
        \item An ineffective risk management framework could have a material adverse effect on our strategic planning and our ability to mitigate risks and/or losses and could have adverse regulatory consequences.
        \item If we fail to retain key employees or recruit new employees, or if we are unable to effectively manage the growth of our employee base, our growth and results of operations could be adversely affected.
    \item If we fail to maintain an effective system of internal control over financial reporting, we may not be able to accurately report our financial results. As a result, current and potential holders of our securities could lose confidence in our financial reporting, which would harm our business and the trading price of our securities.
    \item We rely on quantitative and qualitative models to measure risks and to estimate certain financial values.
    \end{enumerate}
    \item Regulatory and Legal Risks
    \begin{enumerate}
        \item We are subject to extensive regulation that could limit or restrict our activities, impose financial requirements or limitations on the conduct of our business, or result in higher costs to us, and the stringency of the regulatory framework applicable to us may increase if, and as, our balance sheet continues to grow.
        \item As a bank holding company with more than $\$$100 billion of average total consolidated assets, we are subject to stringent regulations, including certain enhanced prudential standards applicable to large bank holding companies. If we exceed certain other thresholds, we will become subject to even more stringent regulations.
        \item We face a risk of noncompliance and enforcement action with the Bank Secrecy Act, other anti-money laundering and anti-bribery statutes and regulations and U.S. economic and trade sanctions.
        \item If we were to violate, or fail to comply with, international, federal or state laws or regulations governing financial institutions, we could be subject to disciplinary action or litigation that could have a material adverse effect on our business, financial condition, results of operations or reputation.
        \item Laws and regulations regarding the handling of personal data and information may impede our services or result in increased costs, legal claims or fines against us.
        \item Adverse results from litigation or governmental or regulatory investigations can impact our business practices and operating results.
        \item A failure to identify and address potential conflicts of interest could adversely affect our businesses.
        \item Anti-takeover provisions and federal laws may prevent a merger or acquisition that may be attractive to stockholders and/or have an adverse effect on our stock price.
        \item Changes in accounting standards could materially impact our financial statements.
        \item We could be adversely affected by changes in tax laws and regulations or their interpretations.
    \end{enumerate}
    \item Strategic and External Risks
    \begin{enumerate}
        \item We face competitive pressures that could adversely affect our business, financial results, or growth.
        \item Our ability to maintain or increase our market share depends on our ability to attract and maintain, as well as meet the needs of, existing and future clients.
        \item We face risks in connection with our strategic undertakings and new business initiatives.
        \item We may fail to realize growth prospects and benefits anticipated as a result of the Boston Private acquisition.
        \item Any damage to our reputation and relationships could have a material adverse effect on our business.
    \end{enumerate}
    \item Other Risks
    \begin{enumerate}
        \item We do not currently pay dividends on shares of our common stock and may not do so in the future.
        \item The price of our capital stock may be volatile or may decline.
        \item Our capital stock is subordinate to our existing and future indebtedness.
    \end{enumerate}
\end{enumerate}

\subsubsection{Summary of WFC's Risk Factors under the Unified Risk Headings}\label{sec:WF_risk_summary2}

\begin{enumerate}
    \item Credit Risks
    \begin{enumerate}
        \item Increased credit risk, including as a result of a deterioration in
        economic conditions or changes in market conditions, could
        require us to increase our provision for credit losses and
        allowance for credit losses and could have a material adverse
        effect on our results of operations and financial condition.
        \item We may have more credit risk and higher credit losses to the
        extent our loans are concentrated by loan type, industry
        segment, borrower type, or location of the borrower or
        collateral.
        \item We may suffer losses, penalties, or other adverse
        consequences if we fail to satisfy our obligations with respect
        to the residential mortgage loans or other assets we originate
        or service.
    \end{enumerate}
    \item Market, Interest Rate, and Liquidity Risks
    \begin{enumerate}
        \item Our financial results have been, and will continue to be, materially affected by general economic conditions, and a deterioration in economic conditions or in the financial markets may materially adversely affect our lending and other
        businesses and our financial results and condition.
        \item Changes in interest rates and financial market values could
        reduce our net interest income and earnings, as well as our
        other comprehensive income, including as a result of
        recognizing losses on the debt and equity securities that we
        hold in our portfolio or trade for our customers.
        \item Effective liquidity management is essential for the operation
        of our business, and our financial results and condition could be
        materially adversely affected if we do not effectively manage
        our liquidity.
        \item Adverse changes in our credit ratings could have a material
        adverse effect on our liquidity, cash flows, and financial results
        and condition.
        \item We rely on dividends from our subsidiaries for liquidity, and
        federal and state law, regulatory requirements, and certain
        contractual arrangements can limit those dividends.
        \item Our mortgage banking revenue can be volatile from quarter to
        quarter, including from the impact of changes in interest rates,
        and we rely on the GSEs to purchase our conforming loans to
        reduce our credit risk and provide liquidity to fund new
        mortgage loans.
    \end{enumerate}
    \item Operational and Technology Risks
    \begin{enumerate}
        \item A failure in or breach of our operational or security systems,
        controls or infrastructure, or those of our third-party vendors
        and other service providers, could disrupt our businesses,
        damage our reputation, increase our costs and cause losses.
        \item A cyberattack or other information security incident could
        have a material adverse effect on our results of operations,
        financial condition, or reputation.
        \item Our framework for managing risks may not be fully effective in
        mitigating risk and loss to us.
        \item Our financial statements require certain assumptions,
        judgments, and estimates and rely on the effectiveness of our
        internal control over financial reporting.
    \end{enumerate}
    \item Regulatory and Legal Risks
    \begin{enumerate}
        \item Current and future legislation and/or regulation could require
        us to change certain of our business practices, reduce our
        revenue and earnings, impose additional costs on us or
        otherwise adversely affect our business operations and/or
        competitive position.
        \item We could be subject to more stringent capital, leverage or
        liquidity requirements or restrictions on our growth, activities
        or operations if regulators determine that our resolution or
        recovery plan is deficient.
        \item Our security holders may suffer losses in a resolution of
        Wells Fargo even if creditors of our subsidiaries are paid in full.
        \item Regulatory rules and requirements may impose higher capital
        and liquidity levels, limiting our ability to pay common stock
        dividends, repurchase our common stock, invest in our
        business, or provide loans or other products and services to our
        customers.
        \item FRB policies, including policies on interest rates, can
        significantly affect business and economic conditions and our
        financial results and condition.
        \item We may be exposed to additional legal or regulatory
        proceedings, costs, and other adverse consequences related to
        instances where customers may have experienced financial
        harm.
        \item We may incur fines, penalties, business restrictions, and other
        adverse consequences from regulatory violations or from any
        failure to meet regulatory standards or expectations.
        \item We are exposed to potential financial loss or other adverse
    \end{enumerate}
    \item Strategic and External Risks
    \begin{enumerate}
        \item Reputational harm, including as a result of our actual or alleged
        conduct or public opinion of the financial services industry
        generally, could adversely affect our business, results of
        operations, and financial condition.
        \item If we are unable to develop and execute effective business
        plans or strategies or manage change effectively, our
        competitive standing and results of operations could suffer.
        \item Our operations and business could be adversely affected by the
        impacts of climate change.
        \item We face significant and increasing competition in the rapidly
        evolving financial services industry.
        \item Our ability to attract and retain qualified employees is critical
        to the success of our business and failure to do so could adversely affect our business performance, competitive
        position and future prospects.
    \end{enumerate}
    \item Other Risks
    \begin{enumerate}
        \item Changes in accounting standards, and changes in how
        accounting standards are interpreted or applied, could
        materially affect our financial results and condition.
    \end{enumerate}
\end{enumerate}

\end{document}